\documentclass[12pt]{article}
\usepackage{amsmath,amsfonts,amsthm,amssymb,times}
\usepackage{setspace}
\usepackage{fancyhdr}
\usepackage{lastpage}
\usepackage{extramarks}
\usepackage{chngpage}
\usepackage{soul,color}
\usepackage{graphicx,float,wrapfig}
\usepackage{listings}
\usepackage{multirow}
\usepackage{tikz}
\usepackage{harpoon}
\usepackage[ruled]{algorithm2e}
\usepackage[hidelinks]{hyperref}
\usepackage{bm}
\usepackage[round]{natbib}
\usepackage{verbatim}
\usepackage{subfigure}
\usetikzlibrary{decorations.pathreplacing}
\usepackage[moderate]{savetrees}
\usepackage{paralist}

\newtheorem{theorem}{Theorem}
\newtheorem{lemma}{Lemma}
\newtheorem{claim}{Claim}
\newtheorem{proposition}{Proposition}
\newtheorem{corollary}{Corollary}
\newtheorem{definition}{Definition}

\newtheorem{assumption}{Assumption}
\newtheorem{remark}{Remark}

\usetikzlibrary{arrows,graphs}

\usepackage[a4paper]{geometry}
\DeclareMathOperator*{\argmax}{arg\,max}

\title{\textbf{Incentivizing Knowledge Transfers}\thanks{We thank David Miller and Harry Di Pei for helpful comments and discussions. All errors are our own. \emph{Email}: Kuang (kuang@ruc.edu.cn); Liu (yi.liu.yl2859@yale.edu); Wei (dwei10@ucsc.edu).}}

\author{\textsc{Zhonghong Kuang}\\ \textit{Renmin U}
\and \textsc{Yi Liu} \\ \textit{Yale}
\and \textsc{Dong Wei}\\ \textit{UC Santa Cruz}}

\date{January 2026}

\begin{document}

\maketitle

\begin{abstract}

We study the optimal design of relational contracts that incentivize an expert to share specialized knowledge with a novice. While the expert fears that a more knowledgeable novice may later erode his future rents, a third-party principal is willing to allocate her resources to facilitate knowledge transfer. 
In the unique profit-maximizing contract between the principal and the expert, the expert is asked to train the novice as much as possible, for free, in the initial period; 
knowledge transfers then proceed gradually and perpetually, while the principal offers lump-sum compensations to the expert right after verifying each transfer; 
even in the long run, a complete knowledge transfer might not be attainable. Our analysis sheds light on the success of several prominent cross-border technology transfers that took place in China's auto industry and Korea's high-speed rail development. 
\bigskip



\end{abstract}

\newpage

\begin{center}
\textit{``Teach an apprentice, and the master may starve."} --- a Chinese proverb.
\end{center}

\section{Introduction}
Knowledge functions as a distinct form of capital---its transmission typically requires relatively few resources but can significantly enhance productivity. However, an individual who possesses specialized knowledge (expert) may be hesitant to share it with a less experienced peer (novice), particularly when doing so could jeopardize their own position or bargaining power. For example, employers often ask senior employees to train junior staff, but a senior employee may fear that, once the junior is fully trained, they could be replaced or see their share of rent diminish as a result.\footnote{This type of concern is common across various professions, such as software engineers, technicians, and lawyers. For documented examples, see Peter Gosselin and Ariana Tobin, ``Cutting `Old Heads' at IBM," \emph{ProPublica}, March 22, 2018; Liz Ryan, ``No -- I Won't Train The New Guy To Replace Me," \textit{Forbes}, August 29, 2017.}
A similar dynamic plays out on a larger scale: Since the 1970s, many developing countries have required multinational corporations to transfer technology to local firms in exchange for market access \citep*{Holmes2015Quid}. Once they acquire the technology, however, domestic firms may use it to compete directly with these multinational companies.\footnote{The joint venture between Shanghai Automotive Industry Corporation and General Motors (SAIC-GM) is a prominent example. As pointed out: ``GM helped rear SAIC into a full-fledged auto maker, with top-tier designers, engineers and marketers." However, ``SAIC could use GM expertise and technology to transform itself into a global auto powerhouse that challenges [GM] down the road." See Sharon Terlep, ``Balancing the Give and Take in GM’s Chinese Partnership," \emph{Wall Street Journal}, August 19, 2012.}

Despite such disincentives, a third party (principal), such as an employer or a local government, is often interested in facilitating knowledge transfer, potentially through leveraging some of its own resources. This is because successful transfers can yield substantial benefits to the principal, including increased firm profitability, higher tax revenues, and/or improved consumer welfare. Given the inherent conflict of interest between the expert and the novice, several natural questions arise: How can a principal effectively incentivize an expert to share knowledge with a novice? What are the key features of the optimal (relational) contract? And what are the implications when an expert eventually retires?

To shed light on these questions, we study a stylized model where a principal (she) interacts over time with two agents, an expert (he) and a novice (he). The expert maintains the maximum knowledge level throughout the game, and the novice begins with a knowledge level of $0$. In each period, the principal makes a payment to the expert, and then the expert trains the novice, which increases the novice's knowledge in the next period to a level at the expert's choice. The stage payoffs depend on the novice's current knowledge level: a more knowledgeable novice leads to higher stage payoffs for both the novice and the principal, but a lower stage payoff for the expert, capturing the potential loss of his task assignment/market share to a more capable novice. In this case, the expert is willing to train the novice only if he will be compensated for the resulting loss of future stage payoffs. 
We also assume that the bilateral surplus \textit{between the principal and the expert} increases with the novice's knowledge level.\footnote{In the baseline model, the novice is a passive player with no decisions to make. In Appendix \ref{sec:micro}, we provide several microfoundations where the novice becomes an active player.}

We model knowledge transfer as an irreversible and non-contractible process with delayed observation: After the novice is trained in a period, the outcome is not observed until the next period; moreover, because knowledge levels and increments are difficult to prove in court, the parties cannot write an enforceable contract that depends on the novice's knowledge level. Consequently, there is a lack of commitment on both sides, making immediate knowledge transfer impossible.\footnote{If knowledge were fully transferred to the novice immediately, the principal would have no incentives to compensate the expert afterwards. It is also impossible to ask the principal to pay in advance because the expert, expecting no further payments, would not have incentives to train the novice.} This dual commitment failure necessitates \emph{gradual} knowledge transfer through self-enforcing \emph{relational contracts} between the principal and the expert, restricted by their respective incentive constraints. Specifically, the expert is willing to train the novice for future payments; meanwhile, the principal prefers to pay the expert in order to continue the transfer process rather than terminating both payments and training.

We fully characterize the unique profit-maximizing contract (Theorem \ref{thm:profit}). In the optimal contract, knowledge is gradually transferred through an infinite-period arrangement.\footnote{Except when the players are highly impatient, in which case no knowledge transfer can be incentivized.} In the first period, the principal asks the expert to transfer as much knowledge as possible to the novice, without paying him right away. After that, in each period, the principal compensates the expert for the present value of his payoff losses incurred from the last period's knowledge transfer; meanwhile, the principal requests the expert to transfer additional knowledge, the present value of which to the principal exactly equals the compensation provided. The optimal knowledge sequence dominates, period by period, all other knowledge sequences that can be implemented by some relational contract, and because of that, it also serves as the unique knowledge sequence in every Pareto-efficient contract (Theorem \ref{thm:pareto}).

To understand these properties, note that the principal pays nothing in the very first period because it does not affect the expert's subsequent incentives to train.\footnote{This result is reminiscent of the ``knowledge gift" in \cite{Garicano2017Relational}, but they occur for very different reasons. See related literature and Section \ref{subsec:rayo} for a discussion.} After the initial period, the principal always fully compensates the expert for the present value of his future losses, \textit{immediately after the training he provided is verified}. In principle, such a compensation could take multiple periods, but doing so would only strengthen the principal's incentives to default in the future, so it is desirable to frontload. Moreover, through a perturbation argument, we show that any optimal contract must bind the expert's incentive constraint in every period, as well as the principal's incentive constraints starting from the second period. These binding constraints pin down a difference equation for knowledge levels in adjacent periods, which allows us to trace out a sequence of knowledge levels from any initial condition. Finally, it turns out that a greater amount of training in the initial period will result in higher knowledge levels in all future periods, so that the principal's optimization problem boils down to incentivizing the highest possible amount of training at the beginning of the game.

Interestingly, the novice's limiting knowledge level in the long run may or may not be the maximum knowledge level. In particular, for moderately patient players, while some knowledge transfers can be sustained, a complete transfer is \textit{not} achievable even in the long run, as any attempt to further increase training would require a payment obligation that is impossible to incentivize the principal to fulfill. As the players become more patient, knowledge transfers are accelerated, and all knowledge will be transferred eventually. Even in the patience limit, however, the transfer process remains gradual and nondegenerate, and it will \textit{not} approach transferring all knowledge immediately (Proposition \ref{prop:delta}).

The qualitative properties of the optimal contract align closely with the \textit{quid pro quo} (``technology for market access") policy adopted in the Chinese automobile industry over the last four decades. A prominent example is the joint venture established in 1997 between Shanghai Automobile Industry Corporation (SAIC) and General Motors (GM), which led to the creation of the SAIC-GM brand for manufacturing passenger vehicles in China. Consistent with our theoretical predictions, the technology transfer process unfolded gradually over many years; meanwhile, the local government provided continuous support to GM, primarily through land price discounts, tax incentives, and preferential procurement treatment. This arrangement ultimately enabled SAIC to launch its own independent automotive brand, Roewe, in 2006, although cars under this local brand were generally regarded as lower in quality compared to those under the SAIC-GM brand. In fact, even by 2012, GM had not yet transferred some of its key technologies. 
The original joint venture agreement also specifically required GM to offer know-how and product lines as equity, thereby pushing for a significant amount of technology transfer at the initial stage of the partnership. Another notable example is the development of Korea Train Express, South Korea's high-speed rail system. The Korean government required transfer of technology in its agreement with Alstom, and the implementation process also seems to match well with our predictions. We discuss these applications in detail in Section \ref{sec:applications}.

In many organizations, the agents (workers) often progress through their career ladders and are not as long-lived as the principal (firm). To account for these features, we extend our analysis to an overlapping-generation model where current experts retire and new novices arrive over time. In the model, each worker lives for $2K$ periods, with the first $K$ periods being his novice phase and the last $K$ periods being the expert phase. The current expert retires and exits the game at the end of his expert phase, at which time the existing novice becomes an expert and a new worker arrives as a novice. This ensures that in every period the principal faces exactly two workers, an expert and a novice, just as in the baseline model. In the event that an upcoming expert is not fully trained by his predecessor, the principal must incur a catch-up cost to qualify him as an expert.

We characterize the profit-maximizing stationary contract and show that it preserves most of the qualitative features from the baseline model, with two substantial differences (Theorem \ref{thm:profit-r}). First, the expert fully trains the novice by the time he retires. This happens because a knowledgeable novice cannot hurt the expert anymore after his retirement. Second, the equilibrium's qualitative properties no longer depend on the discount factor, and in particular, the trivial contract without any transfers is always suboptimal, even for highly impatient players. Although the principal never incurs the catch-up cost in the optimal contract, the presence of such a cost serves as a credible punishment on the principal's default, which helps sustain positive knowledge transfers throughout the relationship.\footnote{If there were no such costs, there would only be positive knowledge transfer in the final period of the expert phase.} Moreover, we find that uniformly increasing the catch-up cost function will benefit the principal and speed up the knowledge transfer process (Proposition \ref{prop:CSgeneralC}). An application to craftsmanship preservation by charitable foundations is discussed.

\medskip

{
\noindent\textbf{Related Literature.} 
This paper contributes broadly to the literature on principal-agent models with relational contracts. Early studies, such as \cite{Bull1987The}, \cite{MacLeod1989Implicit}, and \cite*{Baker1994Subjective}, investigate how firms can induce costly worker effort in repeated games. Subsequent research has explored the impact of other factors on relational contracts, such as information asymmetry \citep{Levin2003Relational,Halac2012Relational,Li2013Managing}, bargaining and renegotiation \citep*{Miller2013A,Miller2023}, and mutual uncertainty about match quality \citep{Kuvalekar2020Job,Kostadinov2022Learning}. In all these models, agents' skill levels are treated as stationary and exogenous. In contrast, we focus on a principal's attempt to facilitate the persistent buildup of a novice's knowledge stock, while abstracting away from effort costs and incomplete information.

Our work also relates to models of gradualism, which have been applied to a variety of contexts, including the prisoners' dilemma \citep{Lockwood2002Gradualism,Fujiwara-Greve2009}, gift exchanges \citep{Kranton1996,Kamada2018Sequential}, hold-up problems \citep{Pitchford2004A,Roy2024}, credit markets \citep{Ghosh2016,Wei2019}, and more broadly, long-term relationships with incomplete information \citep{Watson1999,Hua2022}. Meanwhile, our analysis aims at addressing a central tension in apprenticeships between transferring knowledge and managing future competition, exacerbated by the lack of commitment due to the non-contractibility of knowledge. It also extends to more general environments that involve dynamic career progression.

Most closely related to this paper is the recent line of research on relational knowledge transfers. 
\cite{Garicano2017Relational} examine a setting where the expert relies on a knowledgeable novice to produce output, but the novice can leave at any time with the knowledge he has acquired. They characterize the (expert-)optimal contract between the expert and the novice, in which the novice trades free labor for knowledge over multiple periods. Building on this framework, \cite{Fudenberg2019Training} endogenize effort provision and show that the resulting optimal contract is inefficient: the novice exerts excessive effort and spends too much time on menial work. \cite*{Fudenberg2021Working} further assume that knowledge can only accumulate when the novice exerts a costly effort to work and examine the impact of this learning-by-doing constraint. 

Our paper differs from these studies in three key respects. First, in the aforementioned papers, the novice is the sole productive player whose productivity increases with his knowledge level, so the expert has an intrinsic motivation to train him. In contrast, we highlight a fundamental conflict of interest between the two, arising from the fact that a more capable novice may become a stronger competitor to the expert. Yet, in our setting, it is in the interest of a third-party principal to promote knowledge transfer. Second, in those models, the expert serves as the contract designer who uses knowledge transfers as a main source of reward for the novice's production. In our model, however, the contract designer is the principal who leverages monetary transfers to incentivize the expert to train the novice. Third, \cite{Garicano2017Relational} assume that the expert has one-period commitment, and \cite{Fudenberg2019Training} and \cite*{Fudenberg2021Working} endow the expert with full commitment power. By contrast, our model assumes no commitment on either side, reflecting the non-contractible nature of knowledge. A more detailed comparison between our results and theirs is provided in Section \ref{subsec:rayo}.

}

\section{Model} \label{sec:model}
\subsection{Players, Actions, and Payoffs}
A principal (she) interacts with an expert (he) and a novice (he) over infinite periods $t=0,1,\cdots$, and they discount future payoffs with a common discount factor $\delta\in(0,1)$. The expert possesses one unit of perfectly-divisible productive knowledge, and he can train the novice to increase the novice's knowledge level at any rate he desires. The novice's initial knowledge level at time $0$ is $s_0=0$. Let $s_t\in[0,1]$ be the novice's knowledge level at the beginning of period $t$, and let $x_t\equiv s_{t+1}-s_t$ be the amount of knowledge accrued by the novice in period $t$. The expert's knowledge level remains at $1$ in every period regardless of the training he provides. 

To fix ideas, think of the expert as an experienced worker who engages in both training and production, the novice as an inexperienced worker who only engages in production, and the principal as a firm employing the two. Training is intrinsically \textit{costless} to the expert, but the novice's knowledge level will affect their payoffs by changing the players' relative importance in production. Specifically, given (novice's) knowledge $s_t$ in period $t$, the payoffs (absent monetary transfers) of the principal, expert, and novice within the stage are $\pi(s_t)$, $w(s_t)$, and $v(s_t)$, respectively. Throughout the paper, we maintain the following assumptions on these functions.
\begin{assumption}\label{ass:1}\ 
    \begin{enumerate}
        \item[i)] $\pi$ and $w$ are continuous;
        \item[ii)] $\pi$ and $v$ are strictly increasing, and $w$ is strictly decreasing;
        \item[iii)] $G\equiv \pi+w$ is strictly increasing.
    \end{enumerate}
\end{assumption}
That is, both the principal and the novice strictly benefit from the expert's training of the novice; meanwhile, the more knowledgeable the novice is, the less rent the expert can enjoy in each period, due to the potential loss of his task assignment/market share to a more capable novice. In addition, the bilateral surplus \textit{between the principal and the expert}, $G(s)\equiv \pi(s)+w(s)$, is strictly increasing in $s$. In Appendix \ref{sec:micro}, we provide microfoundations based on various economic settings that generate stage payoff functions satisfying Assumption \ref{ass:1}.

To encourage knowledge transfer, within each period, the principal compensates $p_t$ to the expert; the expert then trains the novice with additional knowledge $x_t$. The stage payoffs of the principal and the expert in period $t$ are given by $\pi(s_t)-p_t$ and $w(s_t)+p_t$, respectively. The timing of the stage game is shown in Figure \ref{fig:timing}.

\begin{figure}[!htb]
\centering
\begin{tikzpicture}[scale=3]
    \draw[line width=1.5pt](0,0.05)--(0,0)--(3,0)--(3,0.05)node[below]at(0,0){$s_t$ is observed}node[above]at(3,0.05){Expert gets $w(s_t)+p_t$}node[above]at(3,0.25){Principal gets $\pi(s_t)-p_t$};
    \draw[dotted,line width=1.5pt](3.2,0)--(3.5,0);
    \draw[line width=1.5pt](3.7,0)--(4.2,0);
    \draw[line width=1.5pt](1,0)--(1,0.05)node[above]{Principal chooses $p_t$};
    \draw[line width=1.5pt](2,0.05)--(2,0)node[below]{Expert chooses $x_t$};
    \draw[line width=1.5pt](3.7,0)--(3.7,0.05)node[below]at(3.7,0){$s_{t+1}$ is observed}node[below]at(3.7,-0.2){($x_t$ is observed)};
    \draw[dotted,line width=1.5pt](-0.5,0)--(-0.2,0);
\end{tikzpicture}
\caption{Timing of the Game.}
\label{fig:timing}
\end{figure}

Because knowledge transfer is difficult to observe immediately, and even more so to contract on, we assume that the principal does not observe $x_t$ until the beginning of $t+1$ and that the players are unable to write an enforceable long-term contract based on any $x_t$. Instead, in period $0$, the principal and the expert agree on a \textbf{relational contract} (or simply, a \textbf{contract}): a self-enforcing agreement that specifies, for each period, a knowledge stock $s_t$ and a compensation $p_t$, conditional on the players staying on path. Let $\bm{s}=(s_0,s_1,\cdots)$ and $\bm{p}=(p_0,p_1,\cdots)$, and we denote a contract by $(\bm{s},\bm{p})$.

Players are risk-neutral and care only about the present value of current and future payoffs. For a given contract $(\bm{s},\bm{p})$, let $\Pi_t(\bm{s},\bm{p})$ and $W_t(\bm{s},\bm{p})$ denote, respectively, the principal's and the expert's \emph{continuation values} from the standpoint of the beginning of period $t$. These continuation values are given by
\begin{align*}
\Pi_t(\bm{s},\bm{p})&=\sum_{\tau=t}^{\infty}\delta^{\tau-t}[\pi(s_\tau)-p_\tau],\\
W_t(\bm{s},\bm{p})&=\sum_{\tau=t}^{\infty}\delta^{\tau-t}[w(s_\tau)+p_\tau].
\end{align*}

\begin{remark}
    The novice in our model is a passive player with no choices to make, and will thus be absent from the analysis. We introduce and make assumptions on his stage payoff, $v$, only for completeness. In the microfoundations in Appendix \ref{sec:micro}, the novice's active decisions will give rise to such an $v$ and also drive the assumptions on $\pi$ and $w$.
\end{remark}

\subsection{Implementable Contracts}
\paragraph{Feasibility.}  A contract $(\bm{s},\bm{p})$ is \textbf{feasible} if
\begin{align}
s_{t+1}\geq s_t,\ s_{t+1}\in [0,1],\quad \forall t\geq 0, \tag{M}\label{eqn:M}\\
p_t\geq 0,\quad \forall t\geq 0.\tag{LL}\label{eqn:LL}
\end{align}
The first condition requires knowledge transfers to be irreversible (i.e., knowledge cannot be ``unlearned"), and the second condition assumes limited liability in monetary transfers.

\paragraph{Incentive Compatibility.} A contract $(\bm{s},\bm{p})$ is \textbf{implementable} if it is feasible and constitutes the on-path plays of a subgame perfect equilibrium (SPE).  

Notice that, starting from any knowledge level $s_t$ at the beginning of time $t$, the continuation strategy profile, $\{p_\tau=0,\ x_\tau=0\}_{\tau\geq t}$ at all histories, forms an SPE in the continuation game, which we call the \textbf{inactive equilibrium}. It turns out that the inactive equilibrium renders the minimum continuation values to both parties, making it a common worst punishment (see Lemma \ref{lem:revelation_principle}). Intuitively, after a missed payment by the principal or an inadequate training by the expert, the two parties agree to stick to the status quo with no further payments or training to be made. This gives rise to the following incentive compatibility constraints. In each period $t$, the principal should weakly prefer continuing the contract to ending all payments and training from this period onward:
\begin{equation}
\Pi_t(\bm{s},\bm{p})\geq \frac{\pi(s_t)}{1-\delta},\quad \forall t\geq 0. \tag{P-IC}\label{eqn:P-IC}
\end{equation}
On the other hand, upon receiving the promised payment $p_t$, the expert should prefer providing the specified training to forfeiting all future payments:
\begin{equation*}
W_t(\bm{s},\bm{p})\geq p_t+\frac{w(s_t)}{1-\delta},\quad \forall t\geq 0. 
\end{equation*}
Because the expert's period-$t$ payoff, $w(s_t)+p_t$, is not affected by his choice in that period, the above condition simplifies to
\begin{equation}
W_{t+1}(\bm{s},\bm{p})\geq \frac{w(s_t)}{1-\delta},\quad \forall t\geq 0. \tag{E-IC}\label{eqn:E-IC}
\end{equation}

\begin{lemma}\label{lem:revelation_principle}
$(\bm{s},\bm{p})$ is implementable if and only if it satisfies \eqref{eqn:M}, \eqref{eqn:LL}, \eqref{eqn:P-IC}, and \eqref{eqn:E-IC}.
\end{lemma}

\subsection{Principal's Problem}
Notice that the \textbf{trivial contract} $\{s_t=s_0,\ p_t=0\}_{t\geq 0}$ is implementable, in which case the principal's time-$0$ profit is $\frac{\pi(s_0)}{1-\delta}$. The principal aims to select an implementable contract that maximizes her time-$0$ profit, which leads to the following program:
\begin{equation}
\max_{(\bm{s},\bm{p})} \ \Pi_0(\bm{s},\bm{p})\quad \mbox{s.t.} \quad \eqref{eqn:M},\ \eqref{eqn:LL},\ \eqref{eqn:P-IC},\ \eqref{eqn:E-IC}. \label{eqn:program}
\end{equation}
\begin{lemma}\label{lem:existence}
    An \textbf{optimal contract} exists. That is, the principal's program \eqref{eqn:program} admits a solution.
\end{lemma}
\begin{remark}
    In this model, we do not impose any participation constraints. One can think of the expert's outside option as valuing no more than $w(1)$ per period, then limited liability ensures that the expert never wants to exit from the relationship.
\end{remark}

\section{Optimal Contracts}\label{sec:optimal}
Below, we first derive some economically interesting necessary conditions that any optimal contract must satisfy, and then we will characterize the unique optimal contract.

\subsection{Necessary Conditions}
\begin{claim}[Nonstop Training]\label{cl:nonstop}
    If a nontrivial contract $(\bm{s},\bm{p})$ is implementable, there does not exist $T$ such that $s_t=s_T$ for all $t\geq T$.  That is, $s_{t+1}>s_t$, infinitely often.
\end{claim}
If training occurs at any point, it must never fully stop. Suppose instead that $x_{T-1}>0$ but no further training is provided starting from period $T$. To be willing to provide $x_{T-1}$, the expert has to be compensated for the loss of his future stage payoffs through $\{p_t\}_{t\geq T}$. However, because $x_t=0$ for all $t\geq T$, the principal would find it optimal not to make any payment from period $T$ onward, rendering the training at $T-1$ impossible. This observation implies that the first-best outcome (i.e., $s_t=1,\ \forall t\geq1$) is not implementable. Even worse, knowledge cannot be fully transferred in finite time.

\begin{claim}\label{cl:binding}
    If a nontrivial contract $(\bm{s},\bm{p})$ is optimal, then $p_0 = 0$, the \eqref{eqn:P-IC} constraints are binding at every $t\geq 1$, and the \eqref{eqn:E-IC} constraints are binding at every $t\geq 0$.
\end{claim}
First, note that $p_0$ does not affect the expert's incentives in any period, but reducing $p_0$ can increase the principal's profit and relax her incentive constraint at $t=0$, so $p_0$ must be zero in any optimal contract. As a result, \eqref{eqn:P-IC} at $t=0$ can be dropped since the principal does not need any incentives to honor a null payment.

Next, if we begin with an implementable $(\bm{s},\bm{p})$ that leaves some of the \eqref{eqn:E-IC} or \eqref{eqn:P-IC} constraints \textit{slack}, as a first step we can always adjust the payment scheme to $\bm{p}^*(\bm{s})$, defined as follows:
\begin{equation}\label{eqn:P*}
\begin{aligned}
    p_0^*&=0,\\
    p_t^* &= \frac{w(s_{t-1})-w(s_t)}{1-\delta},\text{ for }t\geq 1.
\end{aligned}\tag{P*}
\end{equation}
In each period $t$, the expert receives a lump-sum payment equal to the present value of all his payoff losses that arise from the training he provided in the previous period. By construction, the contract $(\bm{s},\ \bm{p}^*(\bm{s}))$ makes all the \eqref{eqn:E-IC} constraints binding.
Because some of the incentive constraints were slack and the discounted total surplus is kept unchanged, the shift to $(\bm{s},\ \bm{p}^*(\bm{s}))$ ensures that \eqref{eqn:P-IC} must be slack in some period $k\geq 1$.

Now consider increasing $s_k$ by a small amount and adjusting the payment scheme accordingly based on \eqref{eqn:P*}. This perturbation does not violate \eqref{eqn:P-IC} in period $k$ because it was made slack by the previous step, and it only relaxes the \eqref{eqn:P-IC} constraints before and after $k$.\footnote{Indeed, other things equal, a higher $s_{k}$ implies a lower $x_{k}$, so the compensation to be made in period $k+1$ becomes lower according to \eqref{eqn:P*}; this relaxes the principal's incentive constraint at $k+1$. In every period prior to $k$, the principal's incentive constraint is also relaxed because her continuation value becomes higher (future surplus gets higher while the expert's continuation value is kept unchanged).} Taken together, whenever some of the \eqref{eqn:E-IC} or \eqref{eqn:P-IC} constraints are slack, this two-step adjustment will deliver an implementable, profit-improving contract, hence the claim.

\bigskip

Claim \ref{cl:binding} implies two tight conditions that any optimal training and payment schemes must satisfy. 
\begin{lemma}[Payment Frontloading]\label{lem:frontloading}
    If $(\bm{s},\bm{p})$ is optimal, then $p_0 = 0$ and $p_t=\frac{w(s_{t-1})-w(s_t)}{1-\delta}$ for every $t\geq 1$. That is, $\bm{p} = \bm{p}^*(\bm{s})$ as defined in \eqref{eqn:P*}.
\end{lemma}

Technically, for a fixed $\bm{s}$, the payment scheme specified in \eqref{eqn:P*} is the \textit{unique} way to make all the \eqref{eqn:E-IC} constraints binding. To see this, note that
\begin{align*}
    p_t &= W_t - w(s_t) - \delta W_{t+1} 
    = \frac{w(s_{t-1})}{1-\delta}  - w(s_t) - \delta\frac{w(s_t)}{1-\delta} 
    = \frac{w(s_{t-1})-w(s_t)}{1-\delta},
\end{align*}
where the first equality follows from the recursive definition of $W_t$ and the second equality is implied by the binding \eqref{eqn:E-IC} constraints.

Economically, this payment scheme \textit{frontloads} the expert's compensation by offering a lump-sum payment for the present value of all his future losses, right after the outcome of his previous training is observed (i.e., $x_{t-1} \equiv s_{t} - s_{t-1}$). To see why this is optimal, note that for each knowledge transfer the expert provides, he has to be compensated for his loss, $w(s_{t-1})-w(s_t)$ per period, through a discounted sum of future payments. 
If the compensation for $x_{t-1}$ were to be made in multiple installments, say, at $t$ and $t+1$, this would imply that \textit{the expert's incentive constraint at $t$ must be slack}, because now the expert's deviation in period $t$ would make him forfeit, at $t+1$, not only the compensation for $x_{t}$ but also part of the compensation for $x_{t-1}$. Consequently, the principal can frontload some payment from period $t+1$ to $t$, which would relax her \eqref{eqn:P-IC} at $t+1$ without violating any other constraints. The relaxation of the principal's incentive constraint then creates room for a strict improvement of the principal's profit (via the perturbation below Claim \ref{cl:binding}).
\bigskip

Adding \eqref{eqn:E-IC} at $t$ to \eqref{eqn:P-IC} at $t+1$, we obtain
\begin{equation}\label{eqn:S-IC}
    \sum_{\tau=t}^{\infty}\delta^{\tau-t}G(s_\tau)\geq \frac{\pi(s_t)+w(s_{t-1})}{1-\delta},\quad \forall t\geq 1.\tag{S-IC}
\end{equation}
\eqref{eqn:S-IC} states that the continuation surplus generated by the contract must outweigh the sum of the players' continuation values in the inactive equilibrium.\footnote{The expert's incentive to train in the current period is provided only by future values. Therefore, it is $w(s_{t-1})$, not $w(s_t)$, that appears on the RHS of \eqref{eqn:S-IC}.} Since the players' incentive constraints are binding in any optimal contract (Claim \ref{cl:binding}), so are the \eqref{eqn:S-IC} constraints. Taking the difference between \eqref{eqn:S-IC} at every $t+1$ and $t$, we have the following result.
\begin{lemma}\label{lem:BE}
    If $(\bm{s},\bm{p})$ is optimal, then $\bm{s}$ satisfies the following \textbf{break-even} condition:  \begin{equation}\label{eqn:recursion}
\delta[\pi(s_{t+1})-\pi(s_t)]=w(s_{t-1})-w(s_t),\quad\forall t\geq1.\tag{BE}
\end{equation}
\end{lemma}
To understand \eqref{eqn:recursion}, recall that the principal needs to pay the expert $\frac{w(s_{t-1})-w(s_t)}{1-\delta}$ in period $t$. Only after this payment is made, the expert will carry out the training of $x_t$, which increases the novice's knowledge from $s_t$ to $s_{t+1}$. From the principal's perspective, this training generates a perpetual benefit of $\pi(s_{t+1})-\pi(s_t)$ each period from $t+1$ onward, which totals $\frac{\delta[\pi(s_{t+1})-\pi(s_t)]}{1-\delta}$. In any optimal contract, condition \eqref{eqn:recursion} requires that, at every $t\geq 1$, the principal's benefit from the expert's training must equal the payment she is supposed to make in that period.
\bigskip

The previous results allow us to further rewrite the principal's objective function. Note that
\begin{align*}
    \Pi_0(\bm{s},\bm{p}) = \pi(s_0) - p_0 + \delta \Pi_1(\bm{s},\bm{p})
    =\pi(s_0)+\frac{\delta\pi(s_1)}{1-\delta},
\end{align*}
where the second equality is implied by the binding \eqref{eqn:P-IC} constraint at $t=1$ and the fact that $p_0=0$. Hence, the principal's problem can be simplified as follows.
\begin{lemma}\label{lem:equivalence}
    A contract $(\bm{s},\bm{p})$ is optimal if and only if $\bm{p}= \bm{p}^*(\bm{s})$ and $\bm{s}$ solves the following program:
    \begin{equation}
        \max_{\bm{s}'}\ s_1' \quad \emph{s.t.}\quad \eqref{eqn:M},\ \eqref{eqn:recursion}.\label{eqn:program2}
    \end{equation}
\end{lemma}

Lemma \ref{lem:equivalence} suggests that the principal should push for substantial knowledge transfer in the initial stage, so long as the scheme remains sustainable. Recall that $p_0^*=0$, so the initial transfer is provided ``for free," or in the language of \cite{Garicano2017Relational}, the expert offers a ``knowledge gift." A profit-maximizing principal should then try to induce the expert to deliver the largest possible knowledge gift.\footnote{We discuss how this prediction is reflected in practice in Section \ref{sec:applications} and compare it to \cite{Garicano2017Relational} in Section \ref{subsec:rayo}.}

Technically, the frontloading payment scheme $\bm{p}^*(\bm{s})$ anchors the expert's time-$0$ value to $\frac{w(s_0)}{1-\delta}$, so the principal only needs to look for an implementable contract that maximizes surplus. Once $s_1$ is fixed, the difference equation given by \eqref{eqn:recursion} fully pins down the entire knowledge sequence. 
It turns out that a greater $s_1$ will lead to a uniformly higher sequence; therefore, it is sufficient to boost $s_1$ (or equivalently, the training in period $0$) as much as possible. 


\subsection{Characterization}
Given any $\{s_t\}_{t\geq 0}$ satisfying \eqref{eqn:recursion}, let the limiting knowledge level be $\bar{s} \equiv \lim_{t}s_t$, and note that feasibility implies $\bar{s}\leq 1$. Summing both sides of \eqref{eqn:recursion} over all $t\geq 1$, we obtain
\begin{equation*}
    \delta[\pi(\bar{s}) - \pi(s_1)] = w(s_0) - w(\bar{s}).
\end{equation*}
This implies the following relationship between $s_1$ and $\bar{s}$:
\begin{equation*}
    \pi(s_1) = \frac{\delta\pi(\bar{s})+w(\bar{s})-w(s_0)}{\delta}.
\end{equation*}
Thus, solving program \eqref{eqn:program2} is equivalent to choosing a limiting knowledge level that maximizes $\delta\pi(\bar{s})+w(\bar{s})$.

Let 
\begin{align}
    \bar{s}^* &\equiv \min\{ \argmax_{s\in [0,1]}\ \delta\pi(s)+w(s)\},\label{eqn:sbar}\\
    s_1^*&\equiv \pi^{-1}\left(\frac{\delta\pi(\bar{s}^*)+w(\bar{s}^*)-w(s_0)}{\delta}\right).\label{eqn:s1}
\end{align}
That is, $\bar{s}^*$ is the \textit{smallest} maximizer of $\delta\pi+w$, and $s_1^*$ is defined using the relationship we just derived. Now we can state the main theorem.
\begin{theorem}\label{thm:profit}
    There exists a unique optimal contract $(\bm{s}^*,\ \bm{p}^*)$. If $\bar{s}^* = 0$, $(\bm{s}^*,\ \bm{p}^*)$ is the trivial contract with no knowledge transfer. If $\bar{s}^* \in (0,\ 1]$, $(\bm{s}^*,\ \bm{p}^*)$ has the following structure:
    \begin{enumerate}
        \item[i)] \textbf{Gradual knowledge transfer:} $s_t^*<s_{t+1}^*$, for all $t\geq 0$;
        \item[ii)] \textbf{Constrained surplus maximization:} $\bm{s}^*$ is the uniformly highest sequence among all that deliver tight incentive constraints, with $s_1^*$ given by \eqref{eqn:s1} and the limiting knowledge level $\bar{s}^*$ given by \eqref{eqn:sbar};
        \item[iii)] \textbf{No upfront payment:} $p_0^* = 0$;
        \item[iv)] \textbf{Frontloaded compensation upon verification of training:} $p_t^* = \frac{w(s_{t-1}^*)-w(s_t^*)}{1-\delta}$.
    \end{enumerate}
\end{theorem}

Figure \ref{fig:optimal} illustrates graphically the process for obtaining the optimal knowledge sequence. Once the limit $\bar{s}^*$ is found as the maximizer of $\delta\pi+w$, $s_1^*$ is determined by equating $\delta\pi(s_1^*)+w(s_0)$ with $\delta\pi(\bar{s}^*)+w(\bar{s}^*)$, and thereafter $s_{t+1}^*$ is determined recursively by equating $\delta\pi(s_{t+1}^*)+w(s_{t}^*)$ with $\delta\pi(\bar{s}^*)+w(\bar{s}^*)$.

\begin{figure}[h!]
\centering
\begin{minipage}[t]{0.9\linewidth}
\centering
\begin{tikzpicture}[scale=5]
\draw[thick](0,0)--(1.1,0)node[right]at(1.1,0){$s$}node[below]at(0,0){$0$}node[below]at(1,0){$1$};
\draw[thick](0,0)--(0,1);
\draw[smooth,domain=0:1]plot(\x,{0.6-0.5*\x^2})node[right]{$w(s)$};
\draw[smooth,domain=0:1]plot(\x,{0.6-0.5*\x^2+0.8*\x})node[right]{$\delta\pi(s)+w(s)$};
\draw[dashed,thick,smooth,domain=0:0.5]plot(\x,{0.6+0.8*\x})node[right]{$\delta\pi(s)+w(s_0)$}node[left]at(0,0.6){$w(s_0)$};
\draw[thick,dotted](0.8,0.92)--(0.8,0)node[below]{$\bar{s}^*$};
\draw[line width=2pt](0,0.92)--(0.8,0.92);
\draw[thick,dotted](0.4,0.92)--(0.4,0)node[below]{$s_1^*$};

\fill (0.8,0.92)circle(0.01);
\fill (0.4,0.92)circle(0.01);
\end{tikzpicture}

(a)
\end{minipage}

\begin{minipage}[t]{0.48\linewidth}
\centering
\begin{tikzpicture}[scale=5]
\draw[thick](0,0)--(1.1,0)node[right]at(1.1,0){$s$}node[below]at(0,0){$0$}node[below]at(1,0){$1$};
\draw[thick](0,0)--(0,1);
\draw[smooth,domain=0:1]plot(\x,{0.6-0.5*\x^2})node[right]{$w(s)$};
\draw[smooth,domain=0:1]plot(\x,{0.6-0.5*\x^2+0.8*\x})node[below]at(1.1,0.9){$\delta\pi(s)+w(s)$};
\draw[thick,dotted](0.8,0.92)--(0.8,0)node[below]{$\bar{s}^*$};
\draw[line width=2pt](0,0.92)--(0.8,0.92);
\draw[thick,dotted](0.4,0.92)--(0.4,0)node[below]{$s_1^*$};
\draw[thick,dotted](0.4,0.52)--(0,0.52);
\draw[dashed,thick,smooth,domain=0:0.6]plot(\x,{0.52+0.8*\x})node[right]{$\delta\pi(s)+w(s_1^*)$}node[left]at(0,0.52){$w(s_1^*)$};
\draw[thick,dotted](0.5,0.92)--(0.5,0)node[below]{$s_2^*$};

\fill (0.8,0.92)circle(0.01);
\fill (0.4,0.92)circle(0.01);
\fill (0.4,0.52)circle(0.01);
\fill (0.5,0.92)circle(0.01);
\end{tikzpicture}

(b)
\end{minipage}
\begin{minipage}[t]{0.48\linewidth}
\centering
\begin{tikzpicture}[scale=5]
\draw[thick](0,0)--(1.1,0)node[right]at(1.1,0){$s$}node[below]at(0,0){$0$}node[below]at(1,0){$1$};
\draw[thick](0,0)--(0,1);
\draw[smooth,domain=0:1]plot(\x,{0.6-0.5*\x^2})node[right]{$w(s)$};
\draw[smooth,domain=0:1]plot(\x,{0.6-0.5*\x^2+0.8*\x})node[below]at(1.1,0.9){$\delta\pi(s)+w(s)$};
\draw[thick,dotted](0.8,0.92)--(0.8,0)node[below]{$\bar{s}^*$};
\draw[line width=2pt](0,0.92)--(0.8,0.92);
\draw[thick,dotted](0.4,0.92)--(0.4,0)node[below]{$s_1^*$};
\draw[thick,dotted](0.5,0.92)--(0.5,0)node[below]{$s_2^*$};
\draw[thick,dotted](0.5,0.475)--(0,0.475);
\draw[dashed,thick,smooth,domain=0:0.65]plot(\x,{0.475+0.8*\x})node[right]{$\delta\pi(s)+w(s_2^*)$}node[left]at(0,0.475){$w(s_2^*)$};
\draw[thick,dotted](0.55625,0.92)--(0.55625,0)node[below]{$s_3^*$};

\fill (0.8,0.92)circle(0.01);
\fill (0.4,0.92)circle(0.01);
\fill (0.4,0.52)circle(0.01);
\fill (0.5,0.92)circle(0.01);
\fill (0.5,0.475)circle(0.01);
\fill (0.55625,0.92)circle(0.01);
\end{tikzpicture}

(c)
\end{minipage}
\caption{Constructing the Optimal Knowledge Sequence for $\pi(s)=s,\ w(s)=0.6-0.5s^2,\ \delta=0.8$.}\label{fig:optimal}
\end{figure}

To understand why $\delta\pi + w$ matters, note that keeping the expert's payoffs constant, the principal always wishes for more training and a more knowledgeable novice. But the expert is willing to train only if he will be compensated after training is provided, and the principal is willing to pay the promised compensation only if she can receive sufficient benefit from \textit{future} training specified in the contract. Therefore, the total amount of future benefits limits the principal's willingness to pay, which then limits how much training the expert wants to provide in each period. 

Given a contract $(\bm{s},\bm{p})$, in any period $k$, sticking to the contract gives the principal roughly $$\frac{\delta \pi(\bar{s})-[w(s_{k-1}) - w(\bar{s})]}{1-\delta},$$
where $\frac{w(s_{k-1})-w(\bar{s})}{1-\delta}$ is the approximate total payment from period $k$ onward until the knowledge level stabilizes, and $\frac{\delta\pi(\bar{s})}{1-\delta}$ is the approximate value of \textit{future} knowledge following the principal's payment at $k$. On the other hand, if the principal were to default, he would only obtain $\frac{\delta\pi(s_{k})}{1-\delta}$ starting $k+1$.\footnote{Note that the principal can get $\pi(s_k)$ at $k$ regardless of his payment decision in the current period, so it drops out of the payoff comparison.} As a result, increasing $\delta\pi(\bar{s})+w(\bar{s})$ would always help enhance the long-run value of the contract to the principal, making higher payments and hence greater knowledge transfers sustainable.

Since $\pi$ is strictly increasing and $w$ is strictly decreasing, it is clear that the discount factor $\delta$ plays a crucial role in shaping the optimal contract.
\begin{proposition}\label{prop:delta}
    There exist  $0\leq \underline{\delta}\leq\bar{\delta}\leq 1$ such that:
    \begin{enumerate}
        \item[i)] the optimal contract is nontrivial if and only if $\delta> \underline{\delta}$;
        \item[ii)] $\bar{s}^*(\delta) < 1$ if $\delta<\bar{\delta}$, and $\bar{s}^*(\delta) = 1$ if $\delta>\bar{\delta}$.
    \end{enumerate}
    Moreover, $\delta_1>\delta_2$ implies $\bm{s}^*(\delta_1)\geq \bm{s}^*(\delta_2)$ uniformly. Finally, $\lim_{\delta\to 1}s_t^*(\delta) < 1$ for every $t$.
\end{proposition}
When the players are highly impatient ($\delta<\underline{\delta}$), the principal can never justify her current payment obligation using the benefit of future training, and without a credible promise of compensation for his future losses, the expert is unwilling to provide any training. As they become moderately patient ($\underline{\delta}<\delta<\bar{\delta}$), some knowledge transfers can be sustained; however, a full transfer is \textit{not} achievable even in the long run, as any attempt to further increase training would require a payment obligation the principal cannot fulfill. When the players care sufficiently about future payoffs ($\delta>\bar{\delta}$), all knowledge will be transferred eventually, but even in the patience limit (as $\delta\to 1$), the transfer process remains nondegenerate with $s_0<s_1^*(1)<s_2^*(1)<\cdots <\bar{s}^*(1) = 1$.

\subsection{Pareto Efficiency}\label{sec:pareto}
The previous analysis focuses on profit maximization, rendering all the bargaining power to the principal in the contract proposal stage. In this section, we show that a weaker notion of Pareto efficiency delivers almost the same sharp characterization.

\begin{definition}
    An implementable contract $(\bm{s},\bm{p})$ is \textbf{Pareto-efficient} if there does not exist another implementable contract $(\bm{s}',\ \bm{p}')$ such that $\Pi_0(\bm{s}',\ \bm{p}')\geq \Pi_0(\bm{s},\bm{p})$ and $W_0(\bm{s}',\ \bm{p}')\geq W_0(\bm{s},\bm{p})$, with at least one strict inequality.
\end{definition}

We start with the following lemma.
\begin{lemma}\label{lem:uniformlyhigh}
    If $(\bm{s},\bm{p})$ is implementable, then $\bm{s}\leq \bm{s}^*$ uniformly, where $\bm{s}^*$ is the optimal knowledge sequence defined by \eqref{eqn:sbar}, \eqref{eqn:s1}, and \eqref{eqn:recursion}.
\end{lemma}

Lemma \ref{lem:uniformlyhigh} indicates that the optimal knowledge sequence $\bm{s}^*$ characterized before uniformly dominates all other implementable knowledge sequences. Note that this is an even stronger result than part (ii) of Theorem \ref{thm:profit} where $\bm{s}^*$ is only shown to be the highest among those that make the incentive constraints binding. 

\begin{theorem}\label{thm:pareto}
    $(\bm{s},\bm{p})$ is Pareto-efficient if and only if $\bm{s}=\bm{s}^*$, $p_0\in \left[0,\frac{\delta[\pi(s_1^*)-\pi(s_0)]}{1-\delta}\right]$, and $p_t = \frac{w(s_{t-1}^*)-w(s_t^*)}{1-\delta}$.
\end{theorem}

\begin{figure}[!htb]
\centering
\begin{minipage}[t]{0.8\linewidth}
\centering
\begin{tikzpicture}[scale=7]
\draw[line width=1.5pt](0,0)--(0.8,0)node[below]at(0.5,0){Expert's time-$0$ value ($W_0$)};
\draw[line width=1.5pt](0,0)--(0,1)node[left]at(0,0.5){\rotatebox{90}{Principal's time-$0$ profit ($\Pi_0$)}};
\fill[fill opacity=0.2](0.1,0.2)--(0.1,0.8)--(0.7,0.2);
\draw[line width=1.5pt] (0.1,0.8) -- (0.7,0.2);
\draw[dotted,line width=1.5pt](0,0.2)--(0.8,0.2)node[right]{$\Pi_0=\frac{\pi(s_0)}{1-\delta}$};
\draw[dotted,line width=1.5pt](0.1,0)--(0.1,1)node[right]{$W_0=\frac{w(s_0)}{1-\delta}$};
\fill[] (0.1,0.8) circle (0.015)node[right]{$A$};
\fill[] (0.7,0.2) circle (0.015)node[above]{$B$};
\draw[]node[right,font=\large]at(0.42,0.5){$\bm{s}=\bm{s}^*$}node[left]at(0.5,0.5){\rotatebox{-45}{Pareto frontier}};
\end{tikzpicture}
\end{minipage}
\caption{Pareto Frontier}\label{fig:sur_tri}
\end{figure}

Theorem \ref{thm:pareto} implies that the entire Pareto frontier can be traced out using the exact characterization of the profit-maximizing contract, with $p_0$ being the surplus splitter. To see why, note first that when $\bm{s} =\bm{s}^*$, all the \eqref{eqn:S-IC} constraints are binding by construction, which implies that the \eqref{eqn:E-IC} constraints are binding all the time. By the same argument as that below Lemma \ref{lem:frontloading}, every $p_t$ from period $1$ onward must satisfy the frontloading formula. Since the \eqref{eqn:P-IC} constraints starting $t=1$ are also ensured by the binding \eqref{eqn:S-IC} constraints, we only need to confirm that $p_0$ does not violate \eqref{eqn:P-IC} at $t=0$, that is,
\begin{equation*}
    p_0 + \pi(s_0) + \delta \Pi_1(\bm{s},\bm{p}) \geq 
    \frac{\pi({s_0})}{1-\delta} \quad \Rightarrow\quad p_0\leq \frac{\delta[\pi(s_1^*)-\pi(s_0)]}{1-\delta}.
\end{equation*}

On the other hand, due to Lemma \ref{lem:uniformlyhigh} and the players' incentive constraints at $t=0$, any implementable contract $(\bm{s},\bm{p})$ must induce a payoff pair within the shaded area of Figure \ref{fig:sur_tri}. It is then clear that as long as $\bm{s}\neq \bm{s}^*$, one can always combine the optimal knowledge sequence with an appropriately chosen $p_0$ to construct a Pareto improvement, hence the theorem.


\begin{remark}
    Our notion of Pareto efficiency considers only the principal and the expert. Theorem \ref{thm:pareto} holds verbatim if the novice's welfare is also included in the definition, because the novice's stage payoff, $v(\cdot)$, is an increasing function (Assumption \ref{ass:1}) and $\bm{s}^*$ uniformly dominates any other implementable $\bm{s}$ (Lemma \ref{lem:uniformlyhigh}).  
\end{remark}

\section{Applications}\label{sec:applications}
In this section, we apply our prior analysis to two cases of large-scale cross-border technology transfers, one in the Chinese auto industry and the other pertaining to the development of Korea Train Express.

\subsection{The Chinese Auto Industry}\label{subsec:auto}

Beginning in the 1970s, many developing countries implemented the \emph{quid pro quo} policy under which multinational corporations were required to transfer technology as a condition for market access \citep*{Holmes2015Quid}. After launching its economic reform in 1978, China urgently needed to learn advanced technologies and explore domestic markets.
Shortly after, China adopted the \emph{quid pro quo} policy in the automobile industry and followed it through until very recently. The requirement for technology transfer was explicit before 2001 (the year China joined the WTO) and became implicit afterwards. \cite*{Bai2022Quid} provide empirical evidence for the effectiveness of this policy and discuss the institutional background in detail. 

The threat to exclude induced many automakers to participate in a technology transfer process, and the primary way to implement it was through joint ventures (JVs) between international automakers and domestic partners.\footnote{In the model, we take participation as given without explicitly considering this pre-game exclusion threat.
}  
The first JV for automobile manufacturing was formed in 1983 between American Motors Corporation and Beijing Jeep Corporation. Other notable examples include the JV between Shanghai Automobile Industry Corporation (SAIC) and Volkswagen established in 1984, and that between SAIC and General Motors (GM) established in 1997. The ``Automotive
Industry Development Policy" issued by China's State Council mandated that cars manufactured by the JVs should reach the international technological level of the same period.\footnote{See Article 16 of Chapter 4 of the \textit{Automotive
Industry Development Policy}.} To meet this requirement, foreign automakers offered know-how and product lines as equity (capped at 50\%) soon after the policy was implemented, while domestic partners provided manufacturing facilities and labor. Consistent with Lemma \ref{lem:equivalence}, this policy helped facilitate substantial knowledge transfer in the initial stage of the partnership.

Establishing JVs was only the first step, as technology transfer was not an overnight process. Instead, it lasted for years via various channels, such as continued training of local workers and management, transfer and joint ownership of patents, and setting up R\&D offices within the JVs. Meanwhile, since these JVs were established, both central and local governments have continually supported participating foreign companies through a variety of subsidies, including guaranteed government procurements, land price discounts, and tax incentives.\footnote{For example, the contract between SAIC and Volkswagen specified that the Shanghai government guaranteed the purchase of Santana vehicles for taxis \citep{Yuan2025The}. Similarly, between 2004 and 2015, most taxis in Beijing (around 80,000 vehicles) were produced by Beijing-Hyundai, a joint venture established in 2002. 

Since 1986, local governments in China have had the authority to lower the costs for land use. Even recently, SAIC-Volkswagen acquired a new land parcel of 95,394 square meters at an ultra-low rate of 1,020 RMB per square meter from the Shanghai Government in 2022.  See Article 4 of the \emph{Provisions of the State Council for the Encouragement of Foreign Investment} 
(English version: \url{http://www.law-lib.com/lawhtm/1986/3910.htm}).

From 1991 to 2007,  eligible foreign enterprises were exempted from income tax during their first two profitable years. Subsequently, they received a 50\% reduction in income tax from the third to the fifth year. See Article 8 of the \emph{Income Tax Law of the People's Republic of China for Enterprises with Foreign Investment and Foreign Enterprises} (English version: \url{http://www.law-lib.com/lawhtm/1991/7536.htm}).
} 
In line with our predictions, the observed technology transfer process featured gradual, subsequent transfers after the initial push by the government; in turn, the government also provided continued compensations to international automakers to keep the scheme running.

However, even in the long run, the technological capacity of domestic manufacturers often remained strictly below that of their foreign partners. Take SAIC as an example. Despite its JVs beginning in the 80s and 90s, SAIC lacked the capability of manufacturing passenger vehicles on its own until 2006, when it launched the inaugural independent automotive brand, Roewe. Even after that, cars sold under this local brand were generally perceived as lower-end and of inferior quality compared to those under a JV brand.

More interestingly, different JVs exhibited distinct levels of collaboration, even after holding fixed the same domestic manufacturer. For example, it was reported that ``GM is holding tight to its more valuable technology: Beijing is eager to tap into foreign auto companies' clean-energy technologies, [b]ut GM doesn't want to share all its research with its Chinese partner."\footnote{See Sharon Terlep, ``Balancing the Give and Take in GM's Chinese Partnership." \textit{Wall Street Journal,} August 19, 2012.} 
In contrast, Volkswagen has appeared more open to collaborations even on its most advanced technologies, working closely with SAIC on developing and deploying innovations regarding electrification and autonomous driving.\footnote{See 
Aschhoff and Ludewig, ``All set for future mobility: Volkswagen Group and SAIC Strengthen Long-Standing Partnership with New Joint Venture Agreement," \textit{Volkswagen Group Press Releases}, November 27, 2024.} 
This divergence may be attributable to the fact that the economic relationship between China and the EU is often considered more stable and less politically sensitive than that between China and the US. As a result, the SAIC-Volkswagen joint venture is arguably associated with a higher discount factor and, as predicted by Proposition \ref{prop:delta}, is able to sustain a greater level of technology transfer in the long run.

\subsection{The Korea Train Express (KTX)}\label{subsec:KTX}

In the 1990s, the government of South Korea (principal) decided to build a high-speed rail (HSR) linking Seoul and Pusan, two of Korea's largest cities. To avoid continued dependence on foreign technologies and its potential jeopardy to operation and profitability, the Korean government specified technology transfer from foreign contractors (expert) to local companies (novice) as part of its request for proposal. 
In 1994, the Franco-British consortium Alstom won the contract after defeating Siemens from Germany and Mitsubishi from Japan, the latter of which was eliminated reportedly because it was unwilling to agree to the technology-transfer clause. For detailed institutional background, see \cite{Kang2014}.

The technology transfer process between Alstom and KTX resembles our characterization in Section \ref{sec:optimal}. First, the contract was structured around strict phasing, featuring gradual technology transfer and gradual payments \citep{Rouach2004Alstom}. In the first phase, Alstom manufactured 12 trainsets in Belfort, France, and delivered them directly. In the second phase, Alstom supplied components, and Korean companies assembled 34 trainsets domestically. In the third phase, Korean companies achieved full localization with Alstom's assistance. 
Meanwhile, payments from the Korean government were also made in a phased manner: Subsequent large-scale procurement funds would only be disbursed after verifying the technology transfer required in each phase, with the government directly testing the trainsets built by local manufacturers.

The Korean government and Alstom also ensured that a substantial amount of technology transfer took place in the initial phase. According to Alstom's press release, ``the technology transfer began with transferring and updating approximately 350,000 high speed train documents: drawings, specifications, manufacturing documents, procedures, purchasing documents, and training documents. Then came the training of Korean engineers to use them: more than 1,000 Koreans were trained in France; over 400 French engineers provided assistance to the Korean production."\footnote{See ``Official opening of the first high speed railway line in Korea," \textit{Alstom Press Releases}, March 30, 2004.
}

In this case, Alstom agreed to transfer all \emph{existing} knowledge to Korea while simultaneously continuing to develop \emph{new} expertise. On one hand, it agreed to fully transfer 100\% of the relevant technologies and permit 95\% of production activities to be carried out locally in Korea. This arrangement enabled Korean firms to acquire critical know-how and hands-on experience, laying the groundwork for them to independently manufacture their own trains in the future. On the other hand, to safeguard its technological advantage, Alstom remained dedicated to conducting efficient research and development initiatives and sustaining the rapid advancement of its core technologies \citep{Rouach2004Alstom}.

It is worth noting that South Korea has officially become an exporter of high-speed rail, competing directly with Alstom in the international market. In June 2024, South Korea and Uzbekistan signed a deal to apply KTX technology in Uzbekistan, with the first batch of trains delivered in December 2025.\footnote{See Nam Hyun-woo, ``Hyundai Rotem delivers trains for Uzbekistan ahead of schedule," \textit{The Korea Times}, December 10, 2025.} This marks the first time KTX technology has been exported, stressing the legit concerns of an expert about the intensified competition that arises from knowledge transfers. 

\section{Expert's Retirement}\label{sec:retire}
A crucial feature of the optimal contract in Section \ref{sec:optimal} is that the knowledge transfer process does \textit{not} end in finite time, and when the players are relatively impatient, knowledge cannot be fully transferred even in the long run. This is somewhat undesirable due to the inefficiencies it causes; moreover, in some applications such as employment relationships, having an expert that is as long-lived as the principal (firm) might not be the most reasonable setting to begin with.

To address these concerns, we consider a slightly different setup with overlapping generations of workers. At any point, the principal interacts with two workers, an expert and a novice. Each worker lives for $2K$ periods. A worker arrives as a novice with knowledge level $s_0=0$. He accumulates knowledge during his novice phase through the training provided by the current expert. The novice phase takes $K$ periods, after which he begins to play the role of an expert for another $K$ periods. Each time a novice is elevated to an expert, the previous expert ``retires" and a new worker joins the relationship as a novice. 

The stage payoffs (absent monetary transfers) of the principal, expert, and novice are still given by $\pi(s)$, $w(s)$, and $v(s)$, respectively, satisfying Assumption \ref{ass:1}, where $s$ is the knowledge level of the novice. Furthermore, if a worker's knowledge level is below $1$ after his novice phase ends (if $s_K<1$), the principal must incur a catch-up cost of $C(s_K)$ to improve it to $1$ so that the worker can start acting as an expert. We assume that $C$ is continuous and strictly decreasing, with $C(1)=0$. This cost captures the extra resources to be expended if a soon-to-be expert is not properly trained by his predecessor. This game is formally specified in Appendix \ref{app:retire}.

We focus on stationary equilibria in which the same contract $(\bm{s},\bm{p})\equiv (s_0,...,s_K;\  p_0,...,p_{K-1})$ is used between the principal and every expert. Because each worker is an inactive player during his novice phase, it suffices to analyze a finite game between the principal and an expert that ends after $K$ periods. 

Given a contract $(\bm{s},\bm{p})$, the continuation values of the principal and the expert at $t\in \{0,1,...,K-1\}$, respectively, are given by
\begin{align*}
\Pi_t^R(\bm{s},\bm{p})&=\sum_{\tau=t}^{K-1}\delta^{\tau-t}[\pi(s_\tau)-p_\tau]-\delta^{K-1-t}C(s_{K}),\\
    W_t^R(\bm{s},\bm{p})&=\sum_{\tau=t}^{K-1}\delta^{\tau-t}[w(s_\tau)+p_\tau].
\end{align*}
Assuming that any deviation from the contract is punished by ending all training and payments that are yet to come (until a new expert appears), we can write down the following incentive compatibility constraints for the principal and the expert:\footnote{This is the common worst punishment between the principal and the \textit{current} expert. One could imagine a more severe punishment to the principal if the next expert continues refusing to train once the training during his novice phase ends prematurely. However, this is arguably less reasonable because 
even though deviations are detectable by a novice, he can only observe the training he receives but not the payments between the other two parties, making him unable to infer who deviated.
}
\begin{align}
\Pi_t^R(\bm{s},\bm{p})&\geq \frac{(1-\delta^{K-t})\pi(s_t)}{1-\delta}-\delta^{K-1-t}C(s_k),\quad \forall t\in\{0,1,...,K-1\};\label{eqn:R-P-IC}\tag{R-P-IC} \\ 
W_{t+1}^R(\bm{s},\bm{p})&\geq \frac{(1-\delta^{K-t-1})w(s_t)}{1-\delta},\quad \forall t\in \{0,...,K-2\}.\label{eqn:R-E-IC}\tag{R-E-IC}
\end{align}
Now the principal's program can be written as:
\begin{equation}
\max_{(\bm{s},\bm{p})} \ \Pi_0^R(\bm{s},\bm{p})\quad \mbox{s.t.} \quad \eqref{eqn:M},\ \eqref{eqn:LL},\ \eqref{eqn:R-P-IC},\ \eqref{eqn:R-E-IC}. \label{eqn:program-r}
\end{equation}
Any solution to \eqref{eqn:program-r} is called an \textbf{optimal contract with retirement}.

\subsection{Analysis}
The optimal contract in this model shares a structure very similar to that in the baseline model, with a substantial difference: knowledge is always fully transferred when an expert retires.

\begin{lemma}\label{lem:sK=1}
    If $(\bm{s},\bm{p})$ is an optimal contract with retirement, then $s_K = 1$.
\end{lemma}
From the current expert's perspective, training the novice no longer has an adverse effect on his future payoffs when he is at the very last period of his tenure, so he is indifferent. Meanwhile, increasing $x_{K-1}$ (thereby increasing $s_K$) reduces the principal's cost of dealing with an unqualified future expert, which improves her continuation value in the contract in every period. This relaxes the principal's incentive constraints and increases her profit. We emphasize that this result holds for every $\delta\in (0,1)$.

In apprenticeships, it is common for a master to teach his student without reservation when he is close to his retirement, even if the student is a potential competitor. Take independent watchmakers as an example:
Under the voluntary guidance of Philippe Dufour and two founders of Greubel Forsey, Michel Boulanger was trained with traditional manual skills and craftsmanship starting in 2007. Later, he put into practice the techniques he had learned by creating a timepiece by hand, which was auctioned for almost \$1.5 million USD to finance the subsequent transmission of craftsmanship to the next generation.\footnote{For more details, see \url{https://timeaeon.org/projects/naissance-dune-montre}.}

The rest of the analysis follows the baseline model closely. By the same reasoning as before, one can show that, in any optimal contract, the expert's incentive constraints must be binding all the time, and so are the principal's incentive constraints starting $t=1$. This implies that the optimal payment sequence must satisfy a modified frontloading formula:
\begin{equation}\label{eqn:R-P*}
\begin{aligned}
    p_0^*&=0,\\
    p_t^* &= \frac{\left(1-\delta^{K-t}\right)[w(s_{t-1})-w(s_t)]}{1-\delta},\text{ for }t\in \{1,...,K-1\}.
\end{aligned}\tag{R-P*}
\end{equation}

In addition, combining the players' binding incentive constraints, we find that the optimal knowledge sequence must satisfy a modified break-even condition: For $t\in \{1,\cdots,K-1\}$,
\vspace{0.25cm}
\begin{equation}\label{eqn:R-BE}
\underbrace{\frac{(\delta-\delta^{K-t})[\pi(s_{t+1})-\pi(s_t)]}{1-\delta}}_\text{Benefit of $x_t$ through stage payoffs}+\underbrace{\delta^{K-1-t}[C(s_t)-C(s_{t+1})]}_\text{Benefit of $x_t$ through cost savings}=\underbrace{\frac{(1-\delta^{K-t})[w(s_{t-1})-w(s_{t})]}{1-\delta}}_\text{Compensation to the expert}.\tag{R-BE}
\vspace{0.25cm}
\end{equation}
Like before, the break-even condition requires that, in each period, the principal's total benefit from the expert's training must equal the payment she is supposed to make in that period. In this case, the expert's training in period $t$ not only affects stage payoffs from period $t+1$ through $K-1$, but also the principal's catch-up cost after the current expert retires.

Then, we can write the principal's time-$0$ profit as
\begin{equation}\label{eqn:Pi0-r}
    \Pi_0^R(\bm{s},\bm{p}) = \pi_0(s_0) - p_0 + \delta\Pi_1^R(\bm{s},\bm{p}) = \pi_0(s_0) + \frac{(\delta - \delta^K)\pi(s_1)}{1-\delta} - \delta^{K-1}C(s_1),
\end{equation}
where the second equality follows from $p_0=0$ and the binding \eqref{eqn:R-P-IC} at $t=1$. The next lemma follows.

\begin{lemma}\label{lem:retire-problem}
    $(\bm{s},\bm{p})$ is an optimal contract with retirement if and only if $\bm{p}$ is given by $\eqref{eqn:R-P*}$ and $\bm{s}$ solves
    \begin{equation}
        \max_{\bm{s}'}\ s_1' \quad \emph{s.t.}\quad \eqref{eqn:M},\ \eqref{eqn:R-BE}.\label{eqn:program2-r}
    \end{equation}
\end{lemma}

Similar to Lemma \ref{lem:equivalence}, the optimal knowledge sequence maximizes training in the very first period (i.e., $x_0$) among those that make the players' incentive constraints binding. Intuitively, the principal's benefit from all the training starting $t=1$ is fully offset by her payment obligations, but the training at $t=0$ is provided ``for free" (i.e., $p_0=0$). Thus, profit maximization boils down to maximizing the principal's rent that results from the initial training. 

\begin{theorem}\label{thm:profit-r}
    There is a unique optimal contract with retirement $(\bm{s}^*,\ \bm{p}^*)$, and it is structured as follows:
    \begin{enumerate}
        \item[i)] \textbf{Gradual knowledge transfer:} $s_0<s_1^*<\cdots<s_{K-1}^*<s_{K}^*=1$;
        \item[ii)] \textbf{Constrained maximization of knowledge gift:} $\bm{s}^*$ generates the highest $s_1^*$ among all that deliver tight incentive constraints;
        \item[iii)] \textbf{No upfront payment:} $p_0^* = 0$;
        \item[iv)] \textbf{Frontloaded compensation upon verification of training:} $p_t^* = \frac{\left(1-\delta^{K-t}\right)[w(s_{t-1}^*)-w(s_t^*)]}{1-\delta}$.
    \end{enumerate}
\end{theorem}
We make a few remarks comparing this characterization to Theorem \ref{thm:profit}. First, with retirement, the trivial contract is never optimal and knowledge is always fully transferred upon an expert's retirement, \textit{regardless of the rate of discounting}. Second, while the optimal knowledge sequence still maximizes knowledge gift and bilateral surplus, it does not necessarily dominate, date by date, other implementable sequences. Third, Pareto efficiency is trickier to define in this more complicated, overlapping-generation model. Even if we restrict attention to the class of stationary equilibria when looking for Pareto dominance, the fact that $\bm{s}^*$ is not always uniformly dominant still leaves room for an implementable $\bm{s}$ that raises the worker's value by sufficiently improving his payoffs during the novice phase.

\subsection{Effect of Catch-Up Cost}
In the model with retirement, the principal has to incur a catch-up cost to qualify a new expert for his role if he is not properly trained by the previous expert. Even though this cost is not incurred in the optimal contract ($s_K^*=1$), it does shape the structure of the contract in an important way.

Absent the catch-up cost, knowledge transfer is not possible except in an expert's final period of the game. This is because the principal would have no incentives to pay in that final period, but then the expert would not want to train the novice in the period before, leading to unraveling. In contrast, with the catch-up cost, if the principal does not pay a retiring expert, the latter can still punish the principal by withholding training in his final period, thereby forcing the catch-up cost onto the principal.  
This credible punishment makes higher payments and better training sustainable from the beginning of the game. The next result formalizes this intuition.

\begin{proposition}\label{prop:CSgeneralC}
Fix any two catch-up cost functions $C(\cdot)$ and $\hat{C}(\cdot)$ such that $\hat{C}(s)>C(s)$ for $s\in[0,1)$. 
\begin{enumerate}
    \item[i)] The principal is strictly better off under $\hat{C}$ than under $C$.
    \item[ii)] For every $t\in \{1,...,K-1\}$, $s_t^*$ converges to $1$ under $\lambda C(\cdot)$ as $\lambda$ increases.
\end{enumerate}
\end{proposition}

Indeed, the optimal contract under $C$ is still implementable under $\hat{C}$, but some of the IC constraints are relaxed. This creates room for inducing greater knowledge transfers and achieving a higher profit. Moreover, when the catch-up cost is sufficiently large, it is possible to frontload almost all training into the first period. Although not much knowledge remains to be transferred, the principal is still willing to pay for such training in the next period, as the expert could eventually punish her by not training the novice fully.

\subsection{Application to Craftsmanship Preservation by Charitable Foundations}

Many traditional crafts are at risk of disappearing because masters (expert) are often reluctant to pass their skills on to outsiders, preferring instead to teach only family members. This preference arises because teaching a family member mitigates the risk of a novice undermining the expert's future earnings: knowledge transfers to a family member may continue to generate benefits for the family, while an outsider has no obligation to support the retired expert. 
However, as noted by the United Nations Educational Scientific and Cultural Organization (UNESCO), ``if family members or community members are not interested in learning it, the knowledge may disappear because sharing it with strangers violates tradition."

To protect these intangible cultural heritages, many grants and foundations (principal) provide monetary incentives to craftsmen to pass on knowledge and skills associated with traditional artisanries to future generations outside of individual families. For example, the Michelangelo Foundation provides opportunities for the younger generation of artisans to collaborate with experienced master craftsmen, engaging in creation and learning together. This pairing ensures that these skills can be passed down from the ``older generation" to the ``new generation." 
UNESCO also raised more than \$4 million USD to support such intergenerational transfers.\footnote{See ``International Assistance from the Intangible Cultural Heritage Fund" (project ID is 199IAS4142).
}

Similar to our characterization, these programs usually facilitate year-long learning of craftsmanship with substantial payments made to the master craftsmen in a gradual manner.  
For example, the Traditional Arts Apprenticeship Grant of the Montana Arts will pay \$3,000 directly to the master, and the payment will be in multiple installments throughout the yearly program. 
The Homo Faber Fellowship, which was jointly developed by the Michelangelo Foundation and Jaeger-LeCoultre (a Swiss luxury watch manufacturer), also provides continuous payments to the master during the 8-month program.

\section{Discussion}\label{sec:discussion}

\subsection{Comparison with \cite{Garicano2017Relational}}\label{subsec:rayo}

In \cite{Garicano2017Relational} (hereafter ``GR"), the expert employs the novice to produce, but the cash-constrained novice initially possesses neither the productive knowledge nor the financial resources to acquire such knowledge. Because of that, the expert is intrinsically motivated to train the novice, but at the same time, he is concerned that the novice may stop producing for him and pursue an outside option whose value increases with the novice's knowledge level. 
GR show that the expert-optimal contract is a multi-period apprenticeship, in which the novice is trained gradually over time and receives full knowledge after a finite number of periods. In the initial period, the expert offers a ``knowledge gift," despite the novice's inability to produce; after that, the novice works without pay until he fully acquires the remaining knowledge; once the apprenticeship concludes, the novice starts to retain all the output from his production.
In what follows, we highlight the key differences between our results and those in GR.




\subsubsection*{Finite vs. Perpetual Transfer Process}

In GR, knowledge is fully transferred after a finite number of periods. This outcome arises because, as the amount of remaining knowledge diminishes, a threshold is eventually reached where the novice's output in a single period exceeds the value of the remaining knowledge. At that point, the expert has an incentive to transfer all remaining knowledge at once to accelerate revenue collection. This result also hinges on the expert's one-period commitment: If the expert were able to freely choose how much knowledge to transfer \textit{after} the novice had produced, he would never transfer all remaining knowledge, knowing that the novice would then appropriate the entire surplus. In contrast, in our model with no commitment, the knowledge transfer process proceeds perpetually, with the novice's knowledge level converging to a long-run limit. When players are somewhat impatient, a complete knowledge transfer is not even attainable at that limit.\footnote{This is true in our baseline model. In the model with retirement, knowledge is always fully transferred upon the expert's retirement, regardless of the discount factor.} If the principal had one-period commitment in our setting, however, the first-best outcome would be achievable: The principal, in period 0, should ask the expert to transfer all knowledge to the novice and commit to paying him $\frac{w(0)-w(1)}{1-\delta}$ in period $1$ upon verifying the transfer.

\subsubsection*{Knowledge Gift}
The optimal contracts in both papers feature a knowledge gift at the beginning of the game, where the expert transfers some knowledge to the novice without receiving any immediate benefit. However, the underlying reasons for this property differ significantly. In GR, the novice is assumed to be unproductive in the initial period, so \textit{by definition}, the expert receives no output in that period regardless of how much knowledge he transfers (which only takes effect in the next period). When deciding the amount of knowledge gift, the expert balances a trade-off between the novice’s productivity and the duration of the apprenticeship. By contrast, in our model, the principal \textit{chooses} not to compensate the expert in the very first period, as a positive payment does not affect the incentives of either party. Instead, the principal pushes the expert to train the novice as much as he can at the outset, so long as the contract remains implementable.

\subsubsection*{Discount Factor and Speed of Knowledge Transfer}
GR and our paper also yield distinct predictions regarding the effect of the discount factor on the speed of knowledge transfer. In GR, greater patience (i.e., a higher discount factor) leads to longer apprenticeships and slower knowledge transfer. This is because an increase in the discount factor raises the marginal value of knowledge to the novice, making him willing to work in exchange for smaller increments of knowledge per period. The expert exploits this by deliberately slowing the transfer of knowledge, thereby extending the apprentice phase during which he captures the full output. To the contrary, in our model, greater patience enhances the principal's value from future knowledge transfers; as a result, the principal is willing to pay more to the expert in each period, which in turn supports larger and faster knowledge transfers. As discussed in Section \ref{subsec:auto}, our prediction seems to align well with the comparison between the long-run cooperation levels observed in the SAIC-GM joint venture and in the SAIC-Volkswagen joint venture.

\subsection{Direct Training Cost}\label{subsubsec:trainingcost}

In the baseline model, training the novice is intrinsically costless to the expert. In this subsection, we consider the situation where the expert also incurs a direct cost of training, denoted by $\underline{c}(s_{t+1}|s_t)$, in order to increase the novice's knowledge level from $s_{t}$ to $s_{t+1}$. We make the following assumption on $\underline{c}$.

\begin{assumption}\label{ass:train}
\ 
\begin{enumerate}
    \item[i)] $\underline{c}(s_t|s_t)=0$, for any $s_t\in[0,1]$;
    \item[ii)] $\underline{c}(s_{t+1}|s_t)$ is increasing in $s_{t+1}$ and decreasing in $s_t$, for any $0\leq s_t\leq s_{t+1}\leq 1$;
    \item[iii)] $\underline{c}(s_{t+1}|s_t)/\partial s_{t+1}\leq \delta G'(s_{t+1})$ almost everywhere. 
\end{enumerate}
\end{assumption}

Parts \textit{i)} and \textit{ii)} of Assumption \ref{ass:train} are basic requirements for a cost function. Meanwhile, part \textit{iii)} limits its derivative, so that the positive effect of greater knowledge transfer on the bilateral surplus will not be fully offset by the increase in training cost.


Under Assumption \ref{ass:train}, one can show that the unique optimal contract exhibits the same qualitative features as what we characterized in Theorem \ref{thm:profit}, except that now the payment in each period compensates the expert for both the direct training cost incurred in the previous period and his loss of future payoffs. Furthermore, uniformly lowering the direct training cost will speed up knowledge transfer (in the sense of increasing $s_1$) and benefit the principal.

\subsection{Alternative Timing}
In the baseline model, we assume delayed observation of the expert's training, capturing that the novice's knowledge growth can be difficult to verify instantly. Meanwhile, the principal's payment in a given period is made and observed \textit{before} the expert makes his decision. In this subsection, we consider an alternative timing under which the principal pays the expert \textit{after} the expert chooses $x_t$, while we still maintain the assumption that $x_t$ is not observed until $t+1$. The new timing is shown in Figure \ref{fig:timing2}.

\begin{figure}[!htb]
\centering
\begin{tikzpicture}[scale=3]
    \draw[line width=1.5pt](0,0.05)--(0,0)--(3,0)--(3,0.05)node[below]at(0,0){$s_t$ is observed}node[above]at(3,0.05){Expert gets $w(s_t)+p_t$}node[above]at(3,0.25){Principal gets $\pi(s_t)-p_t$};
    \draw[dotted,line width=1.5pt](3.2,0)--(3.5,0);
    \draw[line width=1.5pt](3.7,0)--(4.2,0);
    \draw[line width=1.5pt](1,0)--(1,0.05)node[above]{Expert chooses $x_t$};
    \draw[line width=1.5pt](2,0.05)--(2,0)node[below]{Principal chooses $p_t$};
    \draw[line width=1.5pt](3.7,0)--(3.7,0.05)node[below]at(3.7,0){$s_{t+1}$ is observed}node[below]at(3.7,-0.2){($x_t$ is observed)};
    \draw[dotted,line width=1.5pt](-0.5,0)--(-0.2,0);
\end{tikzpicture}
\caption{Alternative Timing of the Game.}
\label{fig:timing2}
\end{figure}

Note that this timing is strategically equivalent to having the principal and expert act \emph{simultaneously} within a period, because neither party can condition its move on the other party's action within the same period. This implies that any consequences to either player's deviation can only occur starting the next period. We have the following characterization of the optimal contract.
\begin{theorem}\label{thm:alternativetiming}
    Let $(\bm{s}^*,\ \bm{p}^*)$ be the optimal contract obtained in Theorem \ref{thm:profit}, with discount factor $\delta^2$. Under the alternative timing, the unique optimal contract $\{(\tilde{x}_t,\ \tilde{p}_t)\}_{t\geq 0}$ has the following structure: For each $k\in \mathbb{N}$,
    \begin{enumerate}
        \item[i)] $\tilde{p}_{2k}=0$ and $\tilde{x}_{2k} = s_{k+1}^*-s_k^*$ in every even period;
        \item[ii)] $\tilde{p}_{2k+1}=(1+\delta)p_{k+1}^*$ and $\tilde{x}_{2k+1} = 0$ in every odd period.
    \end{enumerate}
\end{theorem}

Interestingly, the principal and the expert now have to make alternate moves in order to support the scheme. In every even period (including $t=0$), the expert transfers some knowledge to the novice, but the principal makes no payment. In every odd period, the principal compensates the expert for the present value of his payoff losses due to the knowledge transfer in the period before, but the expert will not provide any further training until the next period.\footnote{The term $(1+\delta)$ appears because $\bm{p}^*$ is obtained using discount factor $\delta^2$ but $\tilde{p}_{2k+1} = \frac{w(s^*_{k})-w(s^*_{k+1})}{1-\delta} = (1+\delta)p^*_{k+1}$.}

To understand this result, note that the optimal payment structure follows from the same reasoning as in the baseline model: Because $p_0$ does not affect any of the expert's incentive constraints, it is optimally set to $0$; moreover, the principal always wants to frontload payment and compensate the expert immediately after verifying each knowledge transfer. 
On the other hand, think from the expert's perspective in period $1$. Because $x_1$ cannot be conditioned on $p_1$ (see Figure \ref{fig:timing2}), $x_1$ can only provide incentives for the prior payment $p_0$, not for the current payment $p_1$. But since $\tilde{p}_0=0$, no such incentives are needed in period $1$, so any positive $x_1$ can be shifted to $x_0$ to accelerate knowledge transfer. That $\tilde{x}_1=0$ further implies that $\tilde{p}_2 = 0$, thereby kick-starting the alternating sequence of moves.

While the optimal contract shares similar qualitative features with the baseline model, the knowledge transfer process is slowed down by a half due to the alternating moves. Furthermore, because the effective discount factor for constructing the optimal contract is now $\delta^2$, the long-run transfer limit also becomes lower (Proposition \ref{prop:delta}). This suggests that the players' sequence of moves in the stage game and how quickly the principal's payment is verified can have a significant effect on the speed and eventual level of knowledge transfer.

\section{Concluding Remarks}

We study how a principal should incentivize an expert to share specialized knowledge with a novice, when the expert is concerned about the future competition he has to face from a more capable novice. In the optimal relational contract between the principal and the expert, the expert is asked to transfer as much knowledge as possible to the novice, without any spot payment, in the initial period. The lack of commitment on both sides necessitates gradual and perpetual knowledge transfers later on, 
with the principal providing lump-sum compensations to the expert for his future losses that arise from each transfer.
When players are somewhat impatient, a complete knowledge transfer is not achievable even in the long run. 

Our characterization matches well with the implementation details of the \textit{quid pro quo} policy in the Chinese auto industry, as well as the development of Korea's high-speed rail. In each of these cases, the government successfully facilitated significant technology transfers from various international companies to domestic manufacturers. Our analysis can also be extended to a more complex environment that accounts for expert retirements and the career progression of novices.

While the model is stylized, we believe that the insights can shed light on a variety of applications, within and beyond the examples discussed in the paper. It would also be interesting to examine the role of other realistic features, such as the novice's learning-by-doing and the 
information asymmetry between different parties. We leave them as promising directions for future research.

\bibliographystyle{aer}
\bibliography{contract}

\newpage
\appendix
\begin{center}
    \Large \textbf{Appendix}
\end{center}
\section{Microfoundations for Stage Payoff Functions}\label{sec:micro}

\subsection{Apprenticeships Within a Firm}

Consider a firm (principal) interacting with two workers, and we call them expert ($i=1$) and novice ($i=2$). There is a divisible task to split between the two workers. The expert always has a knowledge level of $1$, and the novice starts with a knowledge level of $0$.

Within a period, let $\alpha_1$ and $\alpha_2$ be the share of the task assigned to the expert and the novice, respectively ($\alpha_1+\alpha_2\leq 1$). If $\alpha_i$ share of the task is successful, it generates a revenue of $K\alpha_i$ ($0$ if unsuccessful).  Each worker decides whether to exert an effort $e_i\in \{0,1\}$, unobservable to the firm, with cost $e_ic(\alpha_i)$ for some strictly increasing and convex $c(\cdot)$. For each worker, given his knowledge level $s$ and effort $e$, the probability of success is $g(s)[eq+(1-e)p]$ with $1\geq q> p>0$, for some strictly increasing $g(\cdot)$.

Fixing $(\alpha_1,\alpha_2)$, the firm uses a bonus contract to induce high effort. Suppose that a contract pays $b_i$ when the portion of the task controlled by worker $i$ is successful. Exerting effort is optimal for worker $i$, if and only if
\begin{align*}
    &qg(s_i)b_i-c(\alpha_i)\geq pg(s_i)b_i\\
    \Rightarrow\  &b_i \geq \frac{c(\alpha_i)}{g(s_i)(q-p)}.
\end{align*}
Setting $b_i = \frac{c(\alpha_i)}{g(s_i)(q-p)}$, the payoffs of the workers and the firm, respectively, are:
\begin{align*}
    w_i(\alpha_i) &= \frac{pc(\alpha_i)}{q-p},\text{ for }i=1,2;\\
    \pi(\alpha_1,\alpha_2;s) &= K\alpha_1g(1)q -  \frac{qc(\alpha_1)}{q-p} + K\alpha_2g(s)q - \frac{qc(\alpha_2)}{q-p}.
\end{align*}

For simplicity, we assume $c(\alpha) = \frac{\alpha^2}{2}$. To make sure that the firm indeed wants to induce high effort, we also need
\begin{equation}\label{eqn:q>2p}
    q > 2p.
\end{equation}
The firm's task assignment problem within each period is
\begin{align*}
&\max_{\alpha_1,\alpha_2}\ \pi(\alpha_1,\alpha_2;s)\\
\text{s.t. }\ &\alpha_1+\alpha_2\leq 1\\
&0\leq \alpha_i\leq 1,\text{ for }i=1,2.
\end{align*}

Setting up the Lagrangian, we find that the solution, if interior, satisfies
\begin{align*}
    \alpha_1^*(s) &= \frac{1+K(q-p)[g(1)-g(s)]}{2},\\
    \alpha_2^*(s) &= \frac{1-K(q-p)[g(1)-g(s)]}{2}.
\end{align*}
One can see that increasing the novice's knowledge level $s$ will decrease the expert's share of the job, thereby reducing his rent from the relationship. 

Let us further assume $K(q-p) = 1$ and $g(s)=\frac{1+s}{2}$. In this case, the solution is always interior, with $\alpha_1^*(s) = \frac{3-s}{4}$ and $\alpha_2^*(s) = \frac{1+s}{4}$. Substituting them back into the payoff functions, we have
\begin{align*}
    w_1(s) &\equiv w_1(\alpha_1^*(s)) = \frac{Kp}{32}(3-s)^2,\\
    w_2(s) &\equiv w_2(\alpha_2^*(s)) = \frac{Kp}{32}(1+s)^2,\\
    \pi(s) &\equiv \pi(\alpha_1^*(s),\alpha_2^*(s);s) = \frac{Kq}{16}(s^2+2s+9).
\end{align*}
It is straightforward to verify that these functions satisfy Assumption \ref{ass:1}: $\pi$ and $w_2$ are continuous and strictly increasing;\footnote{Here, the novice's payoff function is denoted by $w_2$, not $v$.}  $w_1$ is continuous and strictly decreasing; $G\equiv \pi+w_1$ is strictly increasing whenever $q\geq \frac{3}{2}p$, which is implied by condition \eqref{eqn:q>2p}.

\subsection{Technology Transfers Between Competing Firms}

Consider a government (principal) that wants to bring advanced technologies from foreign companies to its domestic market. Let Firm $1$ be an international firm with full knowledge and Firm 2 be a local firm with an initial knowledge level of $0$.

\subsubsection{Homogeneous Products with a Cost Shifter}

Suppose that the two firms produce the same product and engage in Cournot competition.  The marginal cost of Firm 1 is zero, and the marginal cost of Firm 2 is $1-s$. Within a period, both firms simultaneously choose their quantities, denoted by $q_1$ and $q_2$. The inverse demand curve is $p=A-q_1-q_2$.

The local government wants to induce knowledge transfers from Firm $1$ to Firm $2$. In doing so, it can use monetary transfers (such as tax rebates) to incentivize the international firm. The government also charges a tax on the product, and the exogenous tax rate is $\beta$ per unit.\footnote{The exogenous tax rate is decided by legislation outside of this game.} We assume that the Cournot equilibrium is an interior solution for every $s\in [0,1]$ (this requires $A\geq2+\beta$).

Given $s$, the equilibrium in the stage game is $q_1^*=\frac{A-\beta+(1-s)}{3}$ and $q_2^*=\frac{A-\beta-2(1-s)}{3}$, and their payoffs are 
\begin{align*}
    w(s)&=\left(\frac{A-\beta+(1-s)}{3}\right)^2,\\ v(s)&=\left(\frac{A-\beta-2(1-s)}{3}\right)^2.
\end{align*}
Since the total quantity is $q_1^*+q_2^*=\frac{2(A-\beta)-(1-s)}{3}$, the amount of tax collected within the period is $t(s)=\frac{2(A-\beta)-(1-s)}{3}\beta$. Moreover, the consumer surplus is 
\begin{equation*}
    \mathbf{CS}(s)=\frac{1}{2}\left(\frac{2(A-\beta)-(1-s)}{3}\right)^2.
\end{equation*}
We assume the government cares about consumer welfare and tax revenue, so its stage payoff is
\begin{equation*}
\pi(s)\equiv t(s)+\mathbf{CS}(s).
\end{equation*}

It is easy to verify that $\pi(s)$ and $v(s)$ are increasing in $s$, $w(s)$ is decreasing in $s$, and $G(s)\equiv \pi(s)+w(s)$ is increasing in $s$ whenever $\beta\geq1$; all functions are continuous. So Assumption \ref{ass:1} is satisfied.


\subsubsection{Differentiated Products with a Demand Shifter}

Now suppose that the two firms produce differentiated products and simultaneously set their prices, $p_1$ and $p_2$. The demand for Firm $1$ is $q_1=A+(1-s)+p_2-2p_1$ and the demand for Firm 2 is $q_2=A-2(1-s)+p_1-2p_2$. The production cost is zero, and the exogenous corporate income tax rate is $\gamma$.

Firm 1's profit is $(1-\gamma)q_1p_1=(1-\gamma)[A+(1-s)+p_2-2p_1]p_1$, and its first-order condition is given by $A+(1-s)+p_2-4p_1=0$. Firm 2's profit is $(1-\gamma)q_2 p_2=(1-\gamma)[A-2(1-s)+p_1-2p_2]p_2$, and its first-order condition is given by $A-2(1-s)+p_1-4p_2=0$. 

Assuming interior solutions, we have $p_1^*=\frac{5A+2(1-s)}{15}$ and $p_2^*=\frac{5A-7(1-s)}{15}$. 
So Firm $1$ and Firm $2$'s stage payoffs are
\begin{align*}
w(s)&=\frac{2(1-\gamma)[5A+2(1-s)]^2}{225},\\
v(s)&=\frac{2(1-\gamma)[5A-7(1-s)]^2}{225}.
\end{align*}
The government cares about consumer welfare and tax revenue, so its stage payoff is
\begin{equation*}
\pi(s)\equiv t(s)+\mathbf{CS}(s)
\end{equation*}
where $t(s)=\frac{2\gamma[5A+2(1-s)]^2+2\gamma[5A-7(1-s)]^2}{225}$ and $\textbf{CS}(s)=\frac{100A^2-100A(1-s)+52(1-s)^2}{225}$.\footnote{From the demand system, one can derive the consumer's utility function as $u(q_1,q_2,m)=m+Aq_1+[A-(1-s)]q_2-\frac{q_1^2+q_2^2+q_1q_2}{3}$, where $m$ is the numeraire good.}

One can verify that $v(s)$ is increasing in $s$ and $w(s)$ is decreasing in $s$; $\pi(s)$ and $G(s)\equiv \pi(s)+w(s)$ are increasing in $s$ whenever $A\geq 2$; all functions are continuous. So Assumption \ref{ass:1} is satisfied.

\color{black}

\section{Proofs}

\subsection{Proofs for Sections \ref{sec:model} and \ref{sec:optimal}}\label{app:proof}

\begin{proof}[Proof of Lemma \ref{lem:revelation_principle}]
The \textit{if} part is straightforward because the trigger strategy profile specified in the main text constitutes an SPE, as long as $(\bm{s},\bm{p})$ satisfies \eqref{eqn:M}, \eqref{eqn:LL}, \eqref{eqn:P-IC}, and \eqref{eqn:E-IC}. 

Now we turn to the \textit{only if} part. By feasibility, \eqref{eqn:M} and \eqref{eqn:LL} hold. Notice that, at the beginning of any period $t$, if the current knowledge level is $s_t$, the principal can guarantee herself a continuation value of at least $\frac{\pi(s_t)}{1-\delta}$ by paying zero forever from the current period onward. Likewise, after receiving $p_t$, the expert can guarantee himself a continuation of at least $\frac{w(s_t)}{1-\delta}$ by stopping all training from the current period onward. Hence, if $(\bm{s},\bm{p})$ is supported by an SPE, these particular deviations must be deterred, which implies \eqref{eqn:P-IC} and \eqref{eqn:E-IC}. 
\end{proof}

\begin{proof}[Proof of Lemma \ref{lem:existence}]
It is without loss to bound the principal's payment from above by $\bar{p}\equiv \frac{w(0)-w(1)}{1-\delta}$, as it would have compensated the expert for his largest possible payoff losses at once. With this additional restriction, the domain of program \eqref{eqn:program} is a closed subset of a compact space ($[0,1]^\infty\times[0,\bar{p}]^\infty$, with the product topology), and therefore it is compact. Existence then follows from the extreme value theorem.
\end{proof}

\begin{proof}[Proof of Claim \ref{cl:nonstop}]
The result follows from the argument in the main text.
\end{proof}

\begin{proof}[Proof of Claim \ref{cl:binding}]
Suppose a nontrivial contract $(\bm{s},\bm{p})$ is optimal. That $p_0=0$ follows from the argument in the main text.  

We now show that, for any implementable $(\bm{s},\bm{p})$, if \eqref{eqn:E-IC} is slack in some period $k$, one can then strictly improve the principal's profit in two steps. 

\paragraph{Step 1.} Fixing $\bm{s}$, consider the payment scheme $\bm{p}^*(\bm{s})$ defined in \eqref{eqn:P*}:
\begin{equation*}
\begin{aligned}
    p_0^*&=0,\\
    p_t^* &= \frac{w(s_{t-1})-w(s_t)}{1-\delta},\text{ for }t\geq 1.
\end{aligned}
\end{equation*}
We prove that $(\bm{s},\ \bm{p}^*(\bm{s}))$ is implementable and weakly improves the principal's time-$0$ profit, and that \eqref{eqn:P-IC} is slack in period $k+1$ under $(\bm{s},\ \bm{p}^*(\bm{s}))$.
\medskip

\noindent\textbf{Feasibility.} Since $(\bm{s},\bm{p})$ is feasible, $(\bm{s},\ \bm{p}^*(\bm{s}))$ satisfies \eqref{eqn:M} and  \eqref{eqn:LL}, so it is feasible too.
\medskip

\noindent\textbf{\eqref{eqn:E-IC} constraints are binding.} Recall that $W_t(\bm{s},\ \bm{p}^*(\bm{s}))=\sum_{\tau=t}^{\infty}\delta^{\tau-t} [w(s_\tau)+p_\tau^*]$, so for every $t\geq 0$ we have
\begin{equation*}
W_{t+1}(\bm{s},\ \bm{p}^*(\bm{s}))=\sum_{\tau=t+1}^{\infty}\delta^{\tau-(t+1)} \left[w(s_\tau) + \frac{w(s_{\tau-1})-w(s_{\tau})}{1-\delta}\right] = \frac{w(s_t)}{1-\delta}.
\end{equation*}

\medskip

\noindent\textbf{\eqref{eqn:P-IC} are satisfied.} Because $(\bm{s},\bm{p})$ is implementable, it satisfies \eqref{eqn:P-IC} and \eqref{eqn:E-IC}. Combining these constraints, we have: For every $t\geq 1$
\begin{equation}\label{eqn:inequaity_PE-IC}
\sum_{\tau=t}^{\infty}\delta^{\tau-t}\pi(s_\tau)-\frac{\pi(s_t)}{1-\delta}\geq\sum_{\tau=t}^{\infty}\delta^{\tau-t}p_\tau\geq\frac{w(s_{t-1})}{1-\delta}-\sum_{\tau=t}^{\infty}\delta^{\tau-t}w(s_\tau).
\end{equation}
where the first inequality is \eqref{eqn:P-IC} at $t$ and the second inequality is \eqref{eqn:E-IC} at $t-1$.

Now we check \eqref{eqn:P-IC}. At every $t\geq 1$, 
\begin{align*}
\Pi_t(\bm{s},\ \bm{p}^*(\bm{s}))&=\sum_{\tau=t}^{\infty}\delta^{\tau-t}\pi(s_\tau)-\sum_{\tau=t}^{\infty}\delta^{\tau-t}p_\tau^*\\
&\geq\left[\frac{\pi(s_t)}{1-\delta}+\frac{w(s_{t-1})}{1-\delta}-\sum_{\tau=t}^{\infty}\delta^{\tau-t}w(s_\tau)\right]-\frac{1}{1-\delta}\sum_{\tau=t}^{\infty}\delta^{\tau-t}[w(s_{\tau-1})-w(s_\tau)]\\
&= \frac{\pi(s_t)}{1-\delta},
\end{align*}
where the second line follows condition \eqref{eqn:inequaity_PE-IC} and the definition of $\bm{p}^*(\bm{s})$, and the last line follows from the cancellation between the remaining terms.

At $t=0$,
\begin{equation*}
    \Pi_0(\bm{s},\ \bm{p}^*(\bm{s})) = \pi(s_0) + \delta\Pi_1(\bm{s},\ \bm{p}^*(\bm{s}))\geq \pi(s_0) + \delta\frac{\pi(s_1)}{1-\delta} \geq \frac{\pi(s_0)}{1-\delta},
\end{equation*}
where the first equality follows from $p_0^*=0$, the first inequality follows from \eqref{eqn:P-IC} at $t=1$, and the last inequality follows from $\pi(s_1)\geq \pi(s_0)$.

\medskip

\noindent\textbf{\eqref{eqn:P-IC} is slack in period $k+1$.} Recall that, by assumption, \eqref{eqn:E-IC} is slack in period $k$ under $(\bm{s},\bm{p})$. Then by \eqref{eqn:inequaity_PE-IC}, the following inequality holds:
\[
\sum_{\tau=k+1}^{\infty}\delta^{\tau-(k+1)}\pi(s_\tau)-\frac{\pi(s_{k+1})}{1-\delta}>
\frac{w(s_{k})}{1-\delta}-\sum_{\tau=k+1}^{\infty}\delta^{\tau-(k+1)}w(s_\tau).
\]
We have shown that all (E-IC) constraints are binding under $(\bm{s},\ \bm{p}^*(\bm{s}))$, implying that the RHS above equals $\sum_{\tau=k+1}^{\infty}\delta^{\tau-(k+1)}p_\tau^*$. Therefore, \eqref{eqn:P-IC} is slack in period $k+1$ under $(\bm{s},\ \bm{p}^*(\bm{s}))$.
\medskip

\noindent\textbf{Principal's profit is weakly improved.} This is because the discounted bilateral surplus is solely determined by the knowledge sequence and is thus unchanged, but the payment scheme \eqref{eqn:P*} minimizes the expert's time-$0$ value to $\frac{w(s_0)}{1-\delta}$.


\paragraph{Step 2.} 
We now perturb $\bm{s}$ to construct an implementable $(\bm{s}',\ \bm{p}^*(\bm{s}'))$ which \textit{strictly} increases the principal's time-$0$ profit.

Recall that when the payment scheme \eqref{eqn:P*} is adopted, all the \eqref{eqn:E-IC} constraints are binding. By \eqref{eqn:inequaity_PE-IC}, each \eqref{eqn:P-IC} constraint can be rewritten as \eqref{eqn:S-IC}:
\begin{equation*}
\sum_{\tau=t}^{\infty}\delta^{\tau-t}G(s_\tau)\geq\frac{w(s_{t-1})+\pi(s_t)}{1-\delta}, \ \forall t\geq 1.
\end{equation*}
\medskip

\noindent\textbf{Constructing $\bm{s}'$.} Let $T\geq k+2$ be the first period after $k+1$ that has a higher knowledge level. That is, $s_{k+1} =\cdots =s_{T-1} < s_{T}$. Consider $\bm{s}'$ such that $s_t'=s_t+\varepsilon$ for $t\in \{k+1,\cdots,T-1\}$ and $s_t'=s_t$ otherwise. Because \eqref{eqn:P-IC} is slack at $k+1$ under $(\bm{s},\ \bm{p}^*(\bm{s}))$, by continuity, it still holds under $(\bm{s}',\ \bm{p}^*(\bm{s}'))$ for a small positive $\varepsilon$. In particular, we can set $\varepsilon$ to be the maximum value in $[0,s_{T}-s_{T-1}]$ satisfying 
\[
\sum_{\tau=k+1}^{\infty}\delta^{\tau-k-1}G(s_\tau')\geq\frac{w(s_{k}')+\pi(s_{k+1}')}{1-\delta}.
\]

\medskip

\noindent\textbf{$(\bm{s}',\ \bm{p}^*(\bm{s}'))$ improves the principal's profit.} Since $G(\cdot)$ is strictly increasing, we have $\sum_{\tau=0}^{\infty}\delta^{\tau-t}G(s'_\tau)>\sum_{\tau=0}^{\infty}\delta^{\tau-t}G(s_\tau)$, that is, the discounted bilateral surplus becomes strictly higher. Because the payment scheme \eqref{eqn:P*} anchors the expert's time-$0$ value at $\frac{w(s_0)}{1-\delta}$, the principal's is strictly better off than under the original contract $(\bm{s},\bm{p})$.
\medskip

\noindent\textbf{$(\bm{s}',\ \bm{p}^*(\bm{s}'))$ is implementable.} We only need to check the \eqref{eqn:P-IC} constraints, or equivalently, \eqref{eqn:S-IC}. The RHS of \eqref{eqn:S-IC} is unchanged for $t\in[1,k]\cup[T+1,\infty)$ and the LHS gets weakly higher, so these \eqref{eqn:S-IC} constraints are satisfied. In period $T$, \eqref{eqn:S-IC} holds because
\begin{equation*}
\sum_{\tau=T}^{\infty}\delta^{\tau-T}G(s'_\tau)=\sum_{\tau=T}^{\infty}\delta^{\tau-T}G(s_\tau) \geq \frac{w(s_{T-1})+\pi(s_T)}{1-\delta}>\frac{w(s_{T-1}')+\pi(s_T')}{1-\delta},
\end{equation*}
where the last inequality follows from $s_{T-1}'>s_{T-1}$ and $s_{T}'=s_{T}$. 

If $T=k+2$, we are done verifying all the \eqref{eqn:P-IC} constraints.

If $T>k+2$, then for $k+2\leq t\leq T-1$, we have
\begin{equation*}
\sum_{\tau=t}^{\infty}\delta^{\tau-t}G(s_\tau')\geq \sum_{\tau=t}^{\infty}\delta^{\tau-t}G(s_t')=\frac{\pi(s_t')+w(s_t')}{1-\delta}=\frac{\pi(s_t')+w(s_{t-1}')}{1-\delta},
\end{equation*}
where the inequality holds because $\bm{s}'$ is a nondecreasing sequence and $G(\cdot)$ is increasing, and the last equality follows from $s_{t-1}'=s_t'$.
\end{proof}

\begin{proof}[Proof of Lemma \ref{lem:frontloading}]
    The result follows from the argument in the main text.
\end{proof}

\begin{proof}[Proof of Lemma \ref{lem:BE}]
    If the trivial contract is optimal, then \eqref{eqn:recursion} is automatically satisfied. If a nontrivial contract is optimal, the result follows from the argument in the main text.
\end{proof}

\begin{proof}[Proof of Lemma \ref{lem:equivalence}]
    The result follows from the argument in the main text.
\end{proof}

\begin{proof}[Proof of Theorem \ref{thm:profit}]
Existence follows from Lemma \ref{lem:existence}. By Lemma \ref{lem:equivalence}, any optimal knowledge sequence must be a solution to program \eqref{eqn:program2}. Let the value of program \eqref{eqn:program2} be $\hat{s}_1$. It is then clear from \eqref{eqn:recursion} that $s_0$ and $\hat{s}_1$ will determine a unique optimal knowledge sequence; moreover, given the optimal knowledge sequence, the optimal payment sequence is also uniquely determined by \eqref{eqn:P*}.

Recall that $\bar{s}^*$ and $s_1^*$ are defined in \eqref{eqn:sbar} and \eqref{eqn:s1}.
When $\bar{s}^*=0$, only the trivial knowledge sequence satisfies both \eqref{eqn:M} and \eqref{eqn:recursion}, so by Lemma \ref{lem:equivalence} the trivial contract is optimal.

Now suppose $\bar{s}^*>0$. 
We first argue that $s_1^*$ defined in \eqref{eqn:s1} is indeed the value of program \eqref{eqn:program2} and $\bar{s}^*$ defined in \eqref{eqn:sbar} is indeed the limit of the optimal knowledge sequence. Take any $\bm{s}$ satisfying \eqref{eqn:M} and \eqref{eqn:recursion}, and let $\bar{s} \equiv \lim_{t}s_t$. By the argument in the main text, we know that $s_1$ and $\bar{s}$ satisfy
\begin{equation}\label{eqn:s1sbar}
    \pi(s_1) = \frac{\delta\pi(\bar{s})+w(\bar{s})-w(s_0)}{\delta}.
\end{equation}
Because $\bar{s}^*$ is defined to maximize the RHS of \eqref{eqn:s1sbar} and $\pi$ is strictly increasing, $s_1^*$ is indeed the value of program \eqref{eqn:program2}. Given this $s_1^*$, \eqref{eqn:recursion} then allows us to pin down the entire optimal knowledge sequence $\bm{s}^*$. To see that this sequence converges to $\bar{s}^*$, i.e., the smallest maximizer of $\delta\pi(s)+w(s)$, note first that \eqref{eqn:s1sbar} implies that the limit must be \textit{a} maximizer of \eqref{eqn:s1sbar}, so we have $\lim_ts_t^* \geq \bar{s}^*$. Next, because $\bar{s}^*>0$ and $\pi$ and $w$ are strictly monotone, \eqref{eqn:s1} implies that $s_1^*<\bar{s}^*$. Moreover, given $s_1^*$ defined in \eqref{eqn:s1}, \eqref{eqn:recursion} implies that for every $t\geq 1$,
\begin{equation}\label{eqn:stst-1}
    \pi(s_t^*) = \frac{\delta\pi(\bar{s}^*)+w(\bar{s}^*)-w(s_{t-1}^*)}{\delta}.
\end{equation}
A simple induction tells us that $s_{t}^*<\bar{s}^*$ for every $t$, so $\lim_ts_t^*\leq \bar{s}^*$. Together we have $\lim_ts_t^*= \bar{s}^*$, as desired. Parts \textit{iii)} and \textit{iv)} of Theorem \ref{thm:profit} then follow directly from Lemma \ref{lem:frontloading}. Since $s_1^*>s_0=0$, Part \textit{i)} follows from a simple induction on $\eqref{eqn:recursion}$.

Next, we show that $\bm{s}^*$ uniformly dominates any other knowledge sequence that satisfies \eqref{eqn:M} and \eqref{eqn:recursion}. This is implied by a stronger statement which we now prove: If both $\bm{s}$ and $\bm{s}'$ satisfy \eqref{eqn:M} and \eqref{eqn:recursion} and $s_1'>s_1$, then $\bm{s}'$ uniformly dominates $\bm{s}$. Take any such $\bm{s}$ and $\bm{s}'$ with $s_1'>s_1$. By induction, suppose $s_{t}'>s_{t}$ for some $t\geq 1$. We now show that $s_{t+1}'>s_{t+1}$. Since $\bm{s}$ and $\bm{s}'$ both satisfy \eqref{eqn:recursion}, we have
\begin{align*}
\delta \pi(s_{t+1})+w(s_{t})&=\delta\pi(s_{t})+w(s_{t-1})=\cdots=\delta \pi(s_1)+w(s_0),\\
\delta \pi(s_{t+1}')+w(s_{t}')&=\delta\pi(s_{t}')+w(s_{t-1}')=\cdots=\delta \pi(s_1')+w(s_0).
\end{align*}
Then,
\begin{equation*}
\delta \pi(s_{t+1})+w(s_{t})= \delta \pi(s_1)+w(s_0)<\delta \pi(s_1')+w(s_0)= \delta \pi(s_{t+1}')+w(s_{t}')< \delta \pi(s_{t+1}')+w(s_{t}),
\end{equation*}
where the first inequality holds because $\pi$ is increasing and $s_1'>s_1$, and the last inequality holds because $w$ is decreasing and $s_{t}'>s_{t}$. This immediately implies $s_{t+1}'>s_{t+1}$, as desired.\footnote{We can also show that $\lim_ts_t'>\lim_ts_t$.}
\end{proof}

\begin{proof}[Proof of Proposition \ref{prop:delta}]

Recall that $\bar{s}^*(\delta)$ is defined as the smallest maximizer of $\delta\pi(s)+w(s)$.
Because $\delta\pi(s)+w(s)$ has strictly increasing differences in $(\delta,s)$, we have $\bar{s}^*(\delta')\geq \bar{s}^*(\delta)$ by the Topkis theorem. Then by the maximum theorem, $\bar{s}^*(\delta)$ is left-continuous. Let
\begin{align*}
    \underline{\delta}&\equiv \sup\{\delta\in (0,1)|\bar{s}^*(\delta) = 0\},\\
    \bar{\delta}&\equiv \inf\{\delta\in (0,1)|\bar{s}^*(\delta)=1\}.
\end{align*}
It is easy to verify that they satisfy the desired properties.

Next, we show that $s_t^*(\delta)$ weakly increases with $\delta$ for every $t$. Suppose $\delta'>\delta$. Note that
\begin{align*}
\pi(s_1^*(\delta))&=\pi(\bar{s}^*(\delta))-\frac{w(0)-w(\bar{s}^*(\delta))}{\delta}\\ &\leq\pi(\bar{s}^*(\delta))-\frac{w(0)-w(\bar{s}^*(\delta))}{\delta'} \\
&\leq \frac{\left[\delta'\pi(\bar{s}^*(\delta'))+w(\bar{s}^*(\delta'))\right]-w(0)}{\delta'}\\
&=\pi(s_1^*(\delta')),
\end{align*}
where both the first line and the last line follow from \eqref{eqn:s1}, and the third line follows from $\bar{s}^*(\delta')$ being a maximizer of $\delta'\pi(s)+w(s)$. Hence $s_1^*(\delta')\geq s_1^*(\delta)$. By induction, assume that $s_{t-1}^*(\delta')\geq s_{t-1}^*(\delta)$ for some $t\geq 2$. Then,
\begin{align*}
\pi(s_t^*(\delta))&=\pi(\bar{s}^*(\delta))-\frac{w(s_{t-1}^*(\delta))-w(\bar{s}^*(\delta))}{\delta}\\ 
&\leq\pi(\bar{s}^*(\delta))-\frac{w(s_{t-1}^*(\delta'))-w(\bar{s}^*(\delta))}{\delta'} \\
&\leq \frac{\left[\delta'\pi(\bar{s}^*(\delta'))+w(\bar{s}^*(\delta'))\right]-w(s_{t-1}^*(\delta'))}{\delta'}\\
&=\pi(s_t^*(\delta')),
\end{align*}
where both the first line and the last line follow from \eqref{eqn:stst-1}, the second line follows from $\delta<\delta'$ and $w(s_{t-1}^*(\delta))\geq w(s_{t-1}^*(\delta'))$ (by the induction assumption and monotonicity of $w$), and the third line follows from $\bar{s}^*(\delta')$ being a maximizer of $\delta'\pi(s)+w(s)$. Thus $s_t^*(\delta')\geq s_t^*(\delta)$.

Lastly, note that $\lim_{\delta\to 1}\bar{s}^*(\delta)=1$ because $\pi(s)+w(s)$ is strictly increasing in $s$. Then by \eqref{eqn:s1}, we have
\begin{equation*}
    \lim_{\delta\to 1}s_1^*(\delta)=\lim_{\delta\to 1} \pi^{-1}\left(\frac{\delta\pi(\bar{s}^*)+w(\bar{s}^*)-w(s_0)}{\delta}\right) =\pi^{-1}(\pi(1)+w(1)-w(0))<1.
\end{equation*}
A simple induction on \eqref{eqn:stst-1} shows that $\lim_{t\to 1}s_t^*(\delta)<1$ for every $t$.
\end{proof}

\begin{proof}[Proof of Lemma \ref{lem:uniformlyhigh}]
Take any implementable contract $(\bm{s},\bm{p})$. If $(\bm{s},\bm{p})$ makes all the \eqref{eqn:S-IC} constraints binding, then by part \textit{ii)} of Theorem \ref{thm:profit}, we have $\bm{s}\leq \bm{s}^*$. 

Now suppose that some of the \eqref{eqn:S-IC} constraints are slack under $(\bm{s},\bm{p})$. Applying the two-step construction in the proof of Claim \ref{cl:binding}, we can find another implementable contract $(\bm{s}^{(1)},\ \bm{p}^{(1)})$ such that $\bm{s}^{(1)}\geq\bm{s}$, with strict inequality in some period(s). If $\bm{s}^{(1)}$ satisfies \eqref{eqn:recursion}, we are done because $\bm{s}^*\geq \bm{s}^{(1)}\geq \bm{s}$. If not, we continue this process to construct $\bm{s}^{(2)}$, $\bm{s}^{(3)}$, and so on, and we are done if any $\bm{s}^{(i)}$ along this process satisfies \eqref{eqn:recursion}. 

If no $\bm{s}^{(i)}$ satisfies \eqref{eqn:recursion}, then we obtain a sequence (of sequences), $\{\bm{s}^{(i)}\}_{i=1}^{\infty}$, with $\bm{s}^{(i)}+\varepsilon^{(i)} \bm{l}^{(i)}=  \bm{s}^{(i+1)}$ where $\varepsilon^{(i)}\in(0,1]$ is the incremental adjustment in the two-step construction and $l^{(i)}_t\in\{0,1\}$ is an indicator for such an adjustment at time $t$. Note that $\lim_{i\rightarrow\infty}\varepsilon^{(i)} =0$, because $\sum_{i=1}^{N}\varepsilon^{(i)}\leq\lim_{t\rightarrow\infty}s^{(N)}_{t}\leq 1$ for every $N\geq 1$. Let $\bar{\bm{s}}$ be the pointwise limit of $\bm{s}^{(i)}$, that is, $\bar{s}_t\equiv\lim_{i\rightarrow\infty} s_t^{(i)}$ for every $t$. Because every $\bm{s}^{(i)}$ satisfies \eqref{eqn:S-IC}, so does $\bar{\bm{s}}$. We now argue $\bar{\bm{s}}$ makes all the \eqref{eqn:S-IC} constraints binding, that is, it satisfies \eqref{eqn:recursion}. By contradiction, suppose \eqref{eqn:S-IC}  is slack at time $k$ under $\bar{\bm{s}}$. For simplicity, suppose also that $\bar{s}_{k}<\bar{s}_{k+1}$.\footnote{The case where $\bar{s}_{k}=\bar{s}_{k+1}=\cdots = \bar{s}_{T-1}<\bar{s}_T$ for some $T>k+1$ can be handled with an argument similar to the final part of Claim \ref{cl:binding}'s proof.} In this case, by applying Step 2 in Claim \ref{cl:binding}'s proof, $\bar{s}_{k}$ can be increased by $\bar{\varepsilon} = \min\{s_{T}-s_{T-1},\varepsilon^*\}>0$ where $\varepsilon^*$ satisfies
\[
\sum_{\tau=k}^{\infty}\delta^{\tau-k}G(\bar{s}_\tau)=\frac{w(\bar{s}_{k-1})+\pi(\bar{s}_{k}+\varepsilon^*)}{1-\delta}.
\] 
Recall that 
$\bar{\bm{s}}=\bm{s}^{(N)}+\sum_{i=N}^{\infty}\varepsilon^{(i)}\bm{l}^{(i)}$, so for large enough $N$, $\bm{s}_k^{(N)}$ can be increased by at least $\frac{\bar{\varepsilon}}{2}$. This is a contradiction to $\lim_{i\rightarrow\infty}\varepsilon^{(i)} =0$. Hence, $\bar{\bm{s}}$ satisfies \eqref{eqn:recursion}, and then we have $\bm{s}^*\geq \bar{\bm{s}}\geq \bm{s}$, as desired.
\end{proof}

\begin{proof}[Proof of Theorem \ref{thm:pareto}]
    The result follows from the argument in the main text.
\end{proof}






\subsection{Proofs for Section \ref{sec:retire}}\label{app:retire}

\subsubsection*{Game Setup and Equilibrium Definition}
Consider an overlapping generation model with a long-lived principal and a sequence of workers. At time $0$, there are two workers, an expert and a novice, denoted by $i=-1$ and $i=0$, respectively. Fix an integer $K\geq 2$. At time $iK$, worker $i\geq 1$ arrives as a novice with a knowledge level of $0$; worker $i-2$, who was the expert, exits the game; and worker $i-1$, who was the novice, becomes the expert.\footnote{In other words, every worker $i\geq 0$ lives for $2K$ periods, with the first (last) $K$ periods being his novice (expert) phase; the initial expert, worker $i=-1$, only lives in the first $K$ periods of the game.} In each period $t\in\{iK, iK+1,\cdots,(i+1)K-1\}$, the principal chooses a payment ${p}_t$ and the expert (i.e., worker $i-1$) chooses a training amount ${x}_t$, as illustrated in Figure  \ref{fig:timing}; the novice (i.e., worker $i$) has no actions to take. 

Let $X_t=[0,1]$ and $P_t=\mathbb{R}_+$ be the players' choice sets at time $t$. A (pure, behavioral) strategy of the principal at time $t$ is a mapping 
\begin{equation*}
\sigma_P^{(t)}:X_0\times\cdots\times X_{t-1}\times P_0\times\cdots\times P_{t-1}\rightarrow P_t,
\end{equation*}
which maps all previous training amounts and payments to a current payment. A strategy of expert $i$ at time $t$ is a mapping 
\begin{align*}
\sigma_i^{(t)}&:X_0\times\cdots\times X_{t-1}\times P_0\times\cdots\times P_{t-1}\rightarrow X_t,\quad \text{for } i=-1,\ t\in[0,K),\\
\sigma_i^{(t)}&:X_{iK}\times\cdots\times X_{t-1}\times P_{(i+1)K}\times\cdots\times P_{t-1}\rightarrow X_t,\quad \text{for } i\geq 0,\ t\in[iK,(i+1)K),
\end{align*}
which maps all the training amounts from period $iK$ to $t-1$ and the payments from period $(i+1)K$ to time $t-1$ to a current training amount.\footnote{To be sure, each expert knows the training he received when he was a novice but did not observe the payments between the principal and the previous expert; therefore, $X_{iK}\times\cdots\times X_{(i+1)K-1}$ is in expert $i$'s information set, but $P_{iK}\times\cdots\times P_{(i+1)K-1}$ is not.}

The stage payoffs (absent monetary transfers) of the principal, expert, and novice are still $\pi(s)$, $w(s)$, and $v(s)$, respectively, satisfying Assumption \ref{ass:1}, where $s$ is the knowledge level of the novice. Furthermore, if worker $i$'s knowledge level is below $1$ after his novice phase ends (if $s_{(i+1)K}<1$), the principal must incur a catch-up cost of $C(s_{(i+1)K})$ to improve it to $1$ so that the worker can start acting as an expert. We assume that $C$ is continuous and strictly decreasing, with $C(1)=0$. 

\begin{definition}
    A \textbf{stationary equilibrium}, $\left(\sigma_P,\{\sigma_i\}_{i=-1}^{\infty}\right)$, is an SPE such that
    \begin{enumerate}
        \item[i)] On the equilibrium path, it induces the same contract between the principal and every expert: For every $i\geq -1$, 
        \begin{equation*}(x_{(i+1)K},\cdots,x_{(i+2)K-1};\ p_{(i+1)K},\cdots,p_{(i+2)K-1}) = (x_{0},\cdots,x_{K-1};\ p_{0},\cdots,p_{K-1}).
        \end{equation*}
        \item[ii)] The training received by each expert in his novice's phase does not affect his own training choices: For every $i\geq 0$, $\sigma_i$ is independent of   $X_{iK}\times\cdots X_{(i+1)K-1}$.
        \end{enumerate}
\end{definition}
The second restriction in the above definition precludes a future expert from punishing the principal for deviations that occurred in the past. Even though an expert can detect deviations in training during his novice phase, without observing payments, he is unable to infer who (the principal or the previous expert) the initial deviator was; moreover, whenever needed, the principal is required to pay the catch-up cost to make sure the current expert is properly trained. In this case, it does not seem reasonable for the expert to punish the principal for deviations in previous relationships that were not necessarily the principal's fault, hence the restriction. 

For stationary equilibria, it suffices to analyze a finite game between the principal and an expert that ends after $K$ periods and to represent a contract by $(\bm{s},\bm{p}) \equiv (s_0,\cdots,s_{K};\ p_0,\cdots,p_{K-1})$.


\subsubsection*{Proofs}

\begin{proof}[Proof of Lemma \ref{lem:sK=1}]
The result follows from the argument in the main text.    
\end{proof}

\begin{proof}[Proof of Lemma \ref{lem:retire-problem}]

We first show that in any optimal contract, \eqref{eqn:R-E-IC} must be binding for every $t\in \{0,\cdots,K-2\}$, and \eqref{eqn:R-P-IC} must be binding for every $t\in \{1,\cdots,K-1\}$. 

Suppose by contradiction that $(\bm{s},\bm{p})$ is optimal with retirement and one of these constraints is slack. First consider $(\bm{s},\ \bm{p}^*(\bm{s}))$ where, with an abuse of notation, $\bm{p}^*(\bm{s})$ is given by \eqref{eqn:R-P*}. By the same reasoning as that in the proof of Claim \ref{cl:binding}, $(\bm{s},\ \bm{p}^*(\bm{s}))$ is implementable, makes \eqref{eqn:R-E-IC} binding for every $t\in \{0,\cdots,K-2\}$, and makes \eqref{eqn:R-P-IC} slack for some $t\in \{1,\cdots,K-1\}$. Adding \eqref{eqn:R-P-IC} to \eqref{eqn:R-E-IC} in each period, we have 
\begin{equation*}
\sum_{\tau=t}^{K-1}\delta^{\tau-t}G(s_\tau)\geq \frac{[\pi(s_t)+w(s_{t-1})](1-\delta^{K-t})}{1-\delta}+\delta^{K-t-1}[C(s_K)-C(s_t)],\ \forall t\in \{1,\cdots,K-1\}, \label{eqn:R-S-IC}\tag{R-S-IC}
\end{equation*}
with at least one strict inequality.

Now we further perturb $(\bm{s},\ \bm{p}^*(\bm{s}))$ to construct a strict improvement for the principal. Suppose that \eqref{eqn:R-S-IC} is slack in period $k\in \{1,\cdots,K-1\}$ under $(\bm{s},\ \bm{p}^*(\bm{s}))$. Let $T\geq k+1$ be the first period after $k$ such that $s_T>s_k$ (i.e., $s_T>s_{T-1}=\cdots=s_k$).\footnote{$T$ is well-defined because \textit{i)} by Lemma \ref{lem:sK=1} it is without loss to set $s_K=1$; and \textit{ii)} $s_t<1$ for every $1\leq t\leq K-1$, for otherwise the principal has no incentives to compensate the expert for the training $x_{t-1}$.} Pick any $\varepsilon\in (0,s_{T}-s_{k})$ such that 
\begin{equation*}
 \sum_{\tau=k}^{K-1}\delta^{\tau-k}G(s_\tau)\geq\frac{[\pi(s_k+\varepsilon)+w(s_{k-1})](1-\delta^{K-k})}{1-\delta}+\delta^{K-k-1}[C(s_K)-C(s_k+\varepsilon)].
\end{equation*}

Let $\bm{s}'$ be such that $s_t'=s_t+\varepsilon$ for $t\in\{k,\cdots,T-1\}$ and $s_t'=s_t$ otherwise. We verify that $(\bm{s}',\ \bm{p}^*(\bm{s}'))$ is implementable and strictly improves the principal's profit, which will contradict the optimality of $(\bm{s},\bm{p})$. Profit is strictly improved because $\bm{s}'$ dominates $\bm{s}$ and the expert's time-$0$ value is fixed by \eqref{eqn:R-P*}. For implementability, because \eqref{eqn:R-P*} binds all the \eqref{eqn:R-E-IC} constraints, it suffices to verify \eqref{eqn:R-S-IC} under $(\bm{s}',\ \bm{p}^*(\bm{s}'))$. The RHS of \eqref{eqn:R-S-IC} is unchanged for $t\in[1,k-1]\cup[T+1,K-1]$ and the LHS gets weakly higher, so these \eqref{eqn:R-S-IC} constraints are satisfied. In period $T$, \eqref{eqn:R-S-IC} holds because
\begin{align*}
\sum_{\tau=T}^{K-1}\delta^{\tau-T}G(s'_\tau)=\sum_{\tau=T}^{K-1}\delta^{\tau-T}G(s_\tau)
&\geq \frac{[w(s_{T-1})+\pi(s_T)](1-\delta^{\delta-T})}{1-\delta}+\delta^{K-T-1}[C(s_K)-C(s_T)]\\
&>\frac{[w(s_{T-1}')+\pi(s_T')](1-\delta^{\delta-T})}{1-\delta}+\delta^{K-T-1}[C(s_K')-C(s_T')],
\end{align*}
where the last inequality follows from $s_{T-1}'>s_{T-1}$, $s_{T}'=s_{T}$, and $s_K'=s_K$.

If $T=k+1$, we are done verifying the implementability of $(\bm{s}',\ \bm{p}^*(\bm{s}'))$. If $T>k+1$, then for $t\in[k+1,T-1]$, we have
\begin{align*}
 \sum_{\tau=t}^{K-1}\delta^{\tau-t}G(s_\tau')\geq \sum_{\tau=t}^{K-1}\delta^{\tau-t}G(s_t')&=\frac{[\pi(s_t')+w(s_t')](1-\delta^{K-t})}{1-\delta}\\
 &\geq \frac{[\pi(s_t')+w(s_{t-1}')](1-\delta^{K-t})}{1-\delta}+\delta^{K-t-1}[C(s_K')-C(s_t')],
\end{align*}
where the first inequality holds because $\bm{s}'$ is a nondecreasing sequence, and the last inequality follows from $s_t'=s_{t-1}'$.

Now that we know any optimal contract, $(\bm{s},\bm{p})$, must make \eqref{eqn:R-E-IC} binding for every $t\in \{0,\cdots,K-2\}$ and \eqref{eqn:R-P-IC} binding for every $t\in \{1,\cdots,K-1\}$, note that for every $t\geq 1$,
\begin{align*}
    p_t &= W_{t}^R-w(s_t)-\delta W_{t+1}^R \\
    &= \frac{(1-\delta^{K-t})w(s_{t-1})}{1-\delta} - w(s_t) -  \frac{\delta(1-\delta^{K-t-1})w(s_t)}{1-\delta}\\
    &= \frac{(1-\delta^{K-t})[w(s_{t-1})-w(s_t)]}{1-\delta},
\end{align*}
where the second line follows from the binding \eqref{eqn:R-E-IC} constraints. Moreover, $p_0$ must be zero because otherwise the principal can set it to zero to improve profit without violating any constraints. Thus, $\bm{p}$ satisfies \eqref{eqn:R-P*}. The calculations in the main text show that $\bm{s}$ must be a solution to \eqref{eqn:program2-r}.
\end{proof}

\begin{proof}[Proof of Theorem \ref{thm:profit-r}]
Existence and uniqueness follow from the same argument as that in the proof of Theorem \ref{thm:profit}. Parts \textit{iii)} and \textit{iv)} follow directly from Lemma \ref{lem:retire-problem}. 

For part \textit{i)}, note that if $s_1^*=0$, then \eqref{eqn:R-BE} implies a sequence of zeros; if $s_1^*>0$, then \eqref{eqn:R-BE} implies a strictly increasing sequence. By Lemma \ref{lem:retire-problem}, $s_K^*=1$, so we must have $0<s_1^*<\cdots<s_K^*=1$.

For part \textit{ii)}, take any contract $(\bm{s},\bm{p})$ that makes \eqref{eqn:R-E-IC} binding for every $t\in \{0,\cdots,K-2\}$ and \eqref{eqn:R-P-IC} binding for every $t\in \{1,\cdots,K-1\}$. This implies \eqref{eqn:R-BE} and
\begin{align*}
W_0^R(\bm{s},\bm{p})&=w(s_0)+p_0+\delta W_1^R(\bm{s},\bm{p})\\
&=p_0+w(s_0)+\frac{\delta(1-\delta^{K-1})w(s_0)}{1-\delta}\\
&=p_0+\frac{(1-\delta^K)w(s_0)}{1-\delta}.
\end{align*}
Then, by \eqref{eqn:Pi0-r} and the above equation, the bilateral surplus satisfies 
\begin{align*}
W_0^R(\bm{s},\bm{p}) + \Pi_0^R(\bm{s},\bm{p})= \frac{(1-\delta^K)w(s_0)}{1-\delta}+ \pi_0(s_0) + \frac{(\delta - \delta^K)\pi(s_1)}{1-\delta} - \delta^{K-1}C(s_1).
\end{align*}
Given the monotonicity of $\pi$ and $C$, maximizing bilateral surplus among these contracts is equivalent to solving program \eqref{eqn:program2-r}, as is searching for the optimal contract, so the result follows.
\end{proof}

\begin{proof}[Proof of Proposition \ref{prop:CSgeneralC}]

For part \textit{i)}, let $\bm{s}^*$ and $\hat{\bm{s}}$ denote the optimal knowledge sequences under $C$ and $\hat{C}$, respectively. By their optimality, we have: For all $t\geq 1$,
\begin{align*}
\sum_{\tau=t}^{K-1}\delta^{\tau-t}G(s_\tau^*)
&=\frac{[\pi(s_t^*)+w(s_{t-1}^*)](1-\delta^{K-t})}{1-\delta}-\delta^{K-t-1} C(s_t^*),\\
\sum_{\tau=t}^{K-1}\delta^{\tau-t}G(\hat{s}_\tau)
&=\frac{[\pi(\hat{s}_t)+w(\hat{s}_{t-1})](1-\delta^{K-t})}{1-\delta}-\delta^{K-t-1} \hat{C}(\hat{s}_k).
\end{align*}
Moreover, because $\hat{C}>C$, we also have
\begin{align*}
\sum_{\tau=t}^{K-1}\delta^{\tau-t}G(s_\tau^*)
>\frac{[\pi(s_t^*)+w(s_{t-1}^*)](1-\delta^{K-t})}{1-\delta}-\delta^{K-t-1} \hat{C}(s_t^*).
\end{align*}
This means that $(\bm{s}^*,\ \bm{p}^*(\bm{s}^*))$ is still implementable under $\hat{C}$ and that $\hat{\bm{s}}\neq \bm{s}^*$. Since the optimal contract is always unique, we conclude that the principal is strictly better off under $\hat{C}$ than under $C$.


For part \textit{ii)}, it suffices to show that $\lim_{\lambda\to\infty}s_1^* = 1$. For any $t\in \{1,...,K-1\}$, \eqref{eqn:R-BE} can be written as
\begin{equation*}
    \lambda\delta^{K-t-1}[C(s_t^*)-C(s_{t+1}^*)]=\frac{(1-\delta^{K-t})[w(s_{t-1}^*)-w(s_{t}^*)]}{1-\delta}-\frac{(\delta-\delta^{K-t})[\pi(s_{t+1}^*)-\pi(s_t^*)]}{1-\delta}.
\end{equation*}
Note that the RHS is uniformly bounded, so we must have $\lim_{\lambda\to\infty}(s_{t+1}^*-s_t^*) = 0$. But then,
\begin{equation*}
    \lim_{\lambda\to\infty}s_1^* = \lim_{\lambda\to\infty}\left[s_K^*-\sum_{t=1}^{K-1}(s_{t+1}^*-s_t^*)\right] = 1 - \sum_{t=1}^{K-1}\lim_{\lambda\to\infty}(s_{t+1}^*-s_t^*) = 1,
\end{equation*}
as desired.
\end{proof}

\subsection{Proofs for Section \ref{sec:discussion}}\label{app:discussion}

\begin{proof}[Proof of Theorem \ref{thm:alternativetiming}]


Under the alternative timing, the expert's IC constraints remain unchanged as \eqref{eqn:E-IC}: $W_{t+1}(\bm{s},\bm{p})\geq\frac{w(s_t)}{1-\delta},\ \forall t\geq0$. However, since the expert can only react to the principal's action in the next period, the principal's IC constraints become:
\begin{equation}\label{eqn:A-P-IC}
\Pi_t(\bm{s},\bm{p})\geq\pi(s_t)+\frac{\delta\pi(s_{t+1})}{1-\delta},\quad\forall t\geq0.\tag{P-IC'}
\end{equation}
The principal's problem can now be written as:
\begin{equation}
\max_{(\bm{s},\bm{p})} \ \Pi_0(\bm{s},\bm{p})\quad \mbox{s.t.} \quad \eqref{eqn:M},\ \eqref{eqn:LL},\ \eqref{eqn:A-P-IC},\ \eqref{eqn:E-IC}. \label{eqn:A-program'}
\end{equation}

Following an argument similar to that in the baseline model, one can prove an analogue of Claim \ref{cl:binding}: In any optimal contract, the time-$0$ payment must be zero and all the incentive constraints (except \eqref{eqn:P-IC} at $t=0$) must be binding. This has three implications. First, like before, the optimal payment stream is pinned down by \eqref{eqn:P*}. Second, we obtain a new break-even condition:
\begin{equation}\label{eqn:A-BE}
\delta^2[\pi(s_{t+2})-\pi(s_{t+1})]=w(s_{t-1})-w(s_t),\quad\forall t\geq1.\tag{BE'}
\end{equation}
Third, we can rewrite the principal's profit as
\begin{equation*}
\Pi_0(\bm{s},\bm{p}) = \pi(s_0)+\delta\pi(s_1)+\frac{\delta^2\pi(s_2)}{1-\delta}.
\end{equation*}
Hence, the principal's problem is equivalent to:
\begin{equation}
\max_{\bm{s}} \ \delta\pi(s_1)+\frac{\delta^2\pi(s_2)}{1-\delta}\quad \mbox{s.t.} \quad \eqref{eqn:M},\ \eqref{eqn:A-BE}. \label{eqn:A-program'}
\end{equation}

Let $\tilde{\bm{s}}$ be a solution to \eqref{eqn:A-program'} with discount factor $\delta$, and let $\bm{s}^*$ be the optimal knowledge sequence obtained in Theorem \ref{thm:profit} with discount factor $\delta^2$. Take any sequence $\bm{s}$ satisfying \eqref{eqn:M} and \eqref{eqn:A-BE}, and let $\bar{s} = \lim_ts_t$. Adding up \eqref{eqn:A-BE} across all $t\geq1$, we have
\begin{equation*}
    \delta^2[\pi(\bar{s})-\pi(s_2)]= w(s_0) - w(\bar{s}).
\end{equation*}
This indicates that $\bar{s}$ and $s_2$ mutually pin down each other, and that the choice of $s_1$ does not affect $\bar{s}$. Because the objective function is increasing in both $s_1$ and $s_2$, it is optimal to set $\tilde{s}_1 = \tilde{s}_2 = s_1^*$, where $s_1^*$ is defined by condition \eqref{eqn:s1} using discount factor $\delta^2$.

But then, \eqref{eqn:A-BE} recursively implies that $\tilde{s}_{2k+2} = \tilde{s}_{2k+1}$ for every $k\in \mathbb{N}$, that is, $\tilde{x}_{2k+1} = 0$ in every odd period. Using this fact, one can further turn \eqref{eqn:A-BE} into:
\begin{equation}\label{eqn:14}
    \delta^2\left[\pi(\tilde{s}_{2k+2}) -\pi(\tilde{s}_{2k})\right] = w(\tilde{s}_{2k-2}) - w(\tilde{s}_{2k}),\quad \forall k\geq 1.
\end{equation}
Recall that $\bm{s}^*$ satisfies \eqref{eqn:recursion} with discount factor $\delta
^2$. Comparing \eqref{eqn:14} with \eqref{eqn:recursion}, we conclude that in every even period, we must have $\tilde{s}_{2k} = s^*_k$. Finally, the optimal payment sequence $\tilde{\bm{p}}$ is immediately implied by \eqref{eqn:P*}.
\end{proof}

\end{document}